\documentclass[final,5p]{elsarticle}

\usepackage{hyperref}

\usepackage{amsmath,amssymb,amsfonts}
\usepackage{amstext}
\usepackage{amsthm}
\usepackage{amsbsy}
\usepackage{graphics}
\usepackage{algorithm}
\usepackage[noend]{algorithmic}
\usepackage{mathabx}
\usepackage{multirow}
\usepackage{color}
\usepackage{caption}
\usepackage{epstopdf}
\journal{Signal Processing}

\newtheorem{theorem}{{Theorem}}

\begin{document}

\begin{frontmatter}
\title{Distributed Topology Design for Network Coding Deployed Networks}

%
%
%

\author[rice,baylor]{Minhae Kwon}
\ead{minhae.kwon@rice.edu}
\fntext[label2]{This work was mainly performed while the first author was with Ewha Womans University.}

\author[ewha]{Hyunggon Park\corref{mycorrespondingauthor}}
\cortext[mycorrespondingauthor]{Corresponding author}
\ead{hyunggon.park@ewha.ac.kr}

\address[rice]{Department of Electrical and Computer Engineering, Rice University, Houston, TX, USA}
\address[baylor]{Department of Neuroscience, Baylor College of Medicine, Houston, TX, USA}
\address[ewha]{Department of Electronic and Electrical Engineering, Ewha Womans University, Seoul, Republic of Korea}

\begin{abstract}
In this paper,  we propose a solution to the distributed topology
formation problem for large-scale sensor networks with multi-source multicast flows.
The proposed solution is based on game-theoretic approaches in conjunction with
network coding.
The proposed algorithm requires 
significantly low computational complexity,
while it is known as NP-hard to find an optimal topology for network coding
deployed multi-source multicast flows. In particular, 
we formulate the problem of
distributed network topology formation as a \emph{network formation game} by
considering the nodes  in the network as players that  can take actions  for making
outgoing links. The proposed solution decomposes the original game that 
consists of multiple
players and multicast flows into independent \emph{link formation
games} played by only two players with a unicast flow.   
We also show that the proposed algorithm is guaranteed to determine at least one stable
topology.
Our simulation results confirm that the computational complexity of the
proposed solution is low enough for practical deployment in large-scale networks.
\end{abstract}

\end{frontmatter}


\section{Introduction}


Modern mobile devices can be considered as sophisticated computing and networking platforms with enhanced
sensor capabilities.  
For example, recent mobile medical devices are equipped with significantly
advanced medical sensors that can precisely record the electrocardiography (ECG) signals, body
temperature, blood glucose levels, heart rates, blood oxygen saturation, etc.
By using these sensor data, knowledge extraction, inference and prediction become possible. 
%
%
%
%
%
%
In order for an efficient 
data analysis, the data 
need to be collected in a dedicated storage (e.g., database). 
However, it is not always possible to 
directly deliver data from sources to destinations. Rather, 
the sensor data are often delivered over infrastructureless wireless networks  
with an ad hoc manner. 
With the support of autonomous networking technologies 
(e.g., Qualcomm Wi-Fi SON~\cite{wifison} and Bluetooth mesh networking), 
recent mobile devices can instantaneously form 
large-scale ad hoc networks by making connections among them 
and 
by transmitting the data obtained from 
their sensors and relaying data from other devices. 
Therefore, it is essential to design a stable network topology that can  
simultaneously deliver multiple data flows 
with improved source-to-destination connectivity and network throughput given 
severe power constraints~\cite{li2013blind, li2015energy}.

One of the challenges in designing a network topology is high computational
complexity involved in finding the optimal solution in a large-scale
network~\cite{kar2008, li2014energy, li2018joint}. 
Because recent sensor networks contain a large number of sensor nodes, the number of
potential network topologies increases exponentially with the number of nodes.
Therefore, it is difficult in
general to 
solve the optimization problem for network topology 
unless it is formulated as a special class of optimization problem
(e.g., the convex optimization problem). 
In order to lower the computational complexity associated with finding
optimal network
topology, distributed approaches are often deployed~\cite{mhkwon_ICC17, mhkwonTMC2019, kim2006, komali2008, xu2016}. 

Another challenge in the considered large-scale sensor networks comes 
from multi-source multicast
flows, which are inevitable in ad hoc networks~\cite{chong2003, kwonSPL2017}. 
Multiple sources can be included in the network because sensor nodes can generate
data based on their own sensing operations and deliver information to a set of target
destination nodes, resulting in multicast flows in the network. Examples of
multi-source multicast flows in networks include a sensor grid~\cite{Aalamifar2018}, a
healthcare wireless sensor network~\cite{Johari2018}, and the Internet of
Things~\cite{Deligiannis2017}. 
While multi-source multicast flows frequently overlap in network paths, only one flow
can be delivered at a time.
Such bottleneck paths where the flows are overlapped in a node can incur 
delay in data delivery, resulting in network
throughput degradation~\cite{bellman1953bottleneck, berman1990constrained, punnen1996fast}.
Therefore, the incoming rate of a node should be
taken into account as a constraint for the network topology design problem, such that it
should not exceed the link capacity.
As a solution to the problem of network topology design with the constraint, 
network coding~\cite{Ahlswede2000, chi2008}
can be deployed~\cite{mhkwon_ICC17, mhkwonTMC2019}. 
Network coding is widely known to
have several advantages, such as efficient resource usage (e.g., bandwidth and
power), and improved robustness and throughput~\cite{ Prior2014, Greco2015, cao2014stackelberg, li2014retransmission, mhKWON2014WCNC}.
In this paper, we employ inter-session network coding~\cite{Yang2008} which combines
multiple packets from different sources into a single packet before
transmission~\cite{Chou2007}.
However, 
it is still challenging to design a low-complexity strategy for topology formation
since it is 
an NP-hard
problem to find the optimal network topology in 
a network with multi-source multicast flows, where 
network coding is blindly deployed~\cite{Lehman2004}. 

In this paper, we propose a solution to distributed network topology design that
overcomes the challenges discussed above, while explicitly considering the
multi-source multicast flows in large-scale sensor networks. 
Specifically, we adopt a
game-theoretic approach to formulate a distributed design problem of 
network topology as a
\emph{network formation game}.
The nodes in the network are considered as players in
a game, which can decide to make connections with their neighbor nodes
by considering their utility functions. We design the utility function such that it
represents the 
rewards and costs
associated with making connections. 
In particular, the reward included 
in the utility function
represents the amount of distance reduction towards the destination, 
which can shorten 
the distance from sources to destinations, leading to potential 
network throughput improvement.
Hence, 
the reward of each node can be differently assigned by its location in the network. 
For example, 
a high reward is assigned to a node close to the destination so that 
the node is encouraged to make direct links to the destination, which lead to an 
improved 
source-to-destination connectivity. 
In contrast, a low reward is assigned to a node far from the
destination, such that the node tries to make alternative links by taking advantage 
of 
higher path diversity. 
We impose the cost associated with link formation on the utility function to prevent
nodes from
making redundant outgoing links. Therefore, each node can build connections that can
maximize its utility by explicitly considering the tradeoffs between the rewards and
the cost. 
We claim that the proposed approach is
indeed a distributed solution to network topology formation  because each node 
determines its own action to build links.
Unlike a centralized optimization solution, which must evaluate all possible potential
network topologies and thereby incur high computational complexity, the proposed solution
enables each node to choose its own actions, leading to significantly lower
overall complexity. 
In order to take into account the constraints on the incoming rate, which eventually
leads to a low-complexity solution to the topology design, 
the proposed approach adopts network coding~\cite{Ahlswede2000}. 
Since network coding operations 
combine multiple incoming packets into a single packet, 
the outgoing rate of a node can always be fixed so that it 
can
eliminate the constraints on the incoming rate for a node. 
This enables a node to build its
outgoing links without considering the link formations of other nodes, which means
that decisions about link formation can be made between only two nodes. Therefore, an
$n$-player network formation game that includes multicast flows can be decomposed
into independent $2$-player \emph{link formation games} with a unicast flow, as we
analytically show in this paper. Because the complexity required to solve a
$2$-player link formation game with a unicast flow is significantly lower than that
needed for an $n$-player network formation game with multi-source multicast flow, the
overall complexity can be significantly reduced. Note that if network coding is
not deployed, the $n$-player network formation game cannot be decomposed into $2$-player link formation games.

The main contributions of this paper can be summarized as follows.
\begin{itemize}
\item We formulate the problem of network topology design as a network formation
  game, which leads to a distributed strategy for topology formation.
\item We analytically show that network coding decomposes the network formation game into
link formation games, leading to an algorithm with significantly low complexity.
\item We design a utility function for the network formation game such that nodes can
  explicitly consider the tradeoff between the distance reduction and the cost 
  associated with making links.
\item We quantitatively evaluate the proposed solution and show that the proposed
  solution eventually 
  leads to increased
  network throughput and a reduced number of unnecessary redundant links between
  nodes.
\end{itemize}

Note that the focus of this paper is not on the code design for inter-session network
coding, which has been extensively studied in prior works~\cite{
Kim2009jsac, Khreishah2010, Bourtsoulatze2014TCOM,  Hulya2009, Bourtsoulatze2014TMM,
Douik2016}.  
Rather, we focus on how to design network \emph{topologies} that can lead to improved
network performance (e.g., throughput and delay), which have been 
mostly considered as a given condition in previous literature.

The rest of the paper is organized as follows. In Section~\ref{sec:related_works}, we
briefly review the related works. The network model and detailed process of data
collection and dissemination based on network coding operations are discussed in
Section~\ref{sec:NC_SN}. The network formation game for multicast flows and its
decomposition into link formation games for a unicast flow are proposed in
Section~\ref{sec:NFG}. Simulation results and numerical evaluations are presented in
Section~\ref{sec:simulation}, and conclusions are drawn in
Section~\ref{sec:conclusion}.

\section{Related Works}
\label{sec:related_works}

Before the  notion of network coding,  it was infeasible to achieve an upper bound of
multicast capacity by conventional store-and-forward (SF) relaying
architectures~\cite{PracticalNC03}.  The Steiner tree based topology design can achieve the upper bound of
multicast capacity, but solving the Steiner tree is an NP-hard
problem~\cite{jain2003}. 
In~\cite{Ahlswede2000}, 
it is first shown that network
coding can achieve the maximum throughput via the max-flow min-cut theorem, and 
it is further proved that linear network coding~\cite{Li2003} can achieve the upper bound of capacity. 
The optimal topology solution for a single source scenario is studied in~\cite{Cui2007}; however, in multi-source scenarios, which are frequently observed in sensor network scenarios, the max-flow min-cut bounds cannot
fully characterize the capacity region, and thus, only loose outer bounds~\cite{Yan2006} and sub-optimal solutions~\cite{Traskov2006} are
studied.


Network coding has been deployed in a variety of sensor network
scenarios~\cite{Prior2014, Greco2015,movassaghi2013}. For example, network coding can
improve the energy efficiency of a body area sensor network~\cite{movassaghi2013}. A
robust network coding protocol is proposed for smart grids to enhance the
reliability and speed of data gathering~\cite{Prior2014}. In~\cite{Greco2015}, a
mobile crowd-sensing scenario is considered for decentralized data collection, and
network coding is deployed for energy and spectrum efficiency.

Topology design in sensor networks has been studied in the context of self-organizing
networks~\cite{kim2006,sohrabi2000,Heo2005}: in~\cite{sohrabi2000}, protocols are
proposed for the self-organization of wireless sensor networks with a large number of
mainly static and highly energy constrained nodes; in~\cite{kim2006}, a
self-organizing routing protocol for mobile sensor nodes declares the membership of
a  cluster   as they move and confirms whether a mobile sensor node can communicate
with a specific cluster head within a time slot allocated in a time division multiple
access schedule; in~\cite{Heo2005}, distributed energy efficient deployment
algorithms are proposed 
for mobile sensors and intelligent devices that form an ambient intelligent
network.

Distributed decision making has been widely considered in the field of game theory
and there have been a large number of studies on network formation games not only
in economics but  also  in  engineering~\cite{song2015}. For application to wireless
sensor networks, game-theoretic distributed topology control for wireless
transmission power is proposed in sensor networks~\cite{ komali2008}. The purpose
of topology control is to assign per-node optimal transmission power such that the
resulting topology can guarantee target network connectivity. A similar study of a
topology control game in \cite{eidenbenz2006} aimed to choose the 
optimal power level for
network nodes in ad hoc networks to ensure the desired connectivity properties.
In~\cite{xu2016}, a dynamic topology control scheme that prolongs the lifetime of a
wireless sensor network is provided based on a non-cooperative game.

\begin{table*}[tb]
   \caption{Related studies  in network topology design}
\centering
\begin{tabular}{ |c || c | c| c | c| c | c  |  }
    \hline
 &\cite{kim2006,sohrabi2000}&\cite{komali2008,Heo2005,eidenbenz2006}&\cite{Traskov2006} &\cite{Cui2007}&\cite{xu2016}& This Paper\\  \hline \hline
Source  type&Multi-source&Multi-source&Multi-source&Single source&Multi-source&Multi-source\\  \hline
 Simultaneous 	&&&&&&\\
Multiple&Yes&	Yes	&Yes&	No	&Yes&	Yes\\
Source Support &&&&&&\\
\hline
 Flow type  &Unicast&Unicast&Unicast&Multicast&Multicast&Multicast\\  \hline
 Multiple Point&&&&&&\\ 
Destination   &	No&	No&	No&	Yes	&Yes&	Yes\\
Support &&&&&&\\  \hline
  Solution  type&Centralized&Distributed&Centralized&Distributed&Distributed&Distributed\\  
      \hline
      Network&&&&&&\\        
      Formation  &	High	&Low	&High&	Low&	Low&	Low \\ 
       Complexity &&&&&&\\  
      \hline
Relaying type     &SF&SF&Network coding&Network coding &SF&Network coding \\
  \hline
  \end{tabular}
  \label{tab:previous_works}
 \end{table*}
In Table~\ref{tab:previous_works}, 
several representative related works are classified in terms of source, flow, solution and relaying types. 
In contrast to \cite{kim2006,komali2008,Cui2007,Traskov2006,sohrabi2000,Heo2005,eidenbenz2006}, this paper includes the most generalized source and flow types, i.e., multi-source multicast flows. 
Compared to \cite{xu2016}, which also considers multi-source multicast flows, this paper explicitly considers the network coding function in the topology design problem such that the previously described throughput advantage of network coding in a multicast flow can be properly utilized.

\section{Network Coding Based Sensor Networks  }
\label{sec:NC_SN}

In this section, we describe our network model and network coding based packet dissemination in sensor networks. 
The goal of the network is to deliver all data collected by 
source nodes to their own destination nodes.
In Table~\ref{tab:notation}, a summary of frequently used notation is presented. 
\begin{table*}[bt]
 \caption{Summary of notations}
\centering
\begin{tabular}{ c c  |c c  }
    \hline
    Notation & Description &  Notation & Description \\
     \hline
$\mathcal G$ & a direct graph & $\mathcal V(\mathcal G)$&a set of node in $\mathcal G$ \\
$\mathcal E(\mathcal G)$& a set of directed links in $\mathcal G$ & $\mathcal L_n$ & a subgraph\\
$v_i$ & a node with index $i$ & $x_i$ & source data collected by $v_i$\\
$e_{ij}$ & a directed link from $v_i$ to $v_j$ & $\delta_{ij}$ & Euclidean distance between $v_i$ and $v_j$\\
$X_{ij}$ & data transmitted via $e_{ij}$ & $J_{\mathcal G}$ & network status of $\mathcal G$\\
$\Phi$ & network coding operation & $\mathbf G_{\mathcal G}(\mathbf D)$ & network formation game of $\mathcal G$ with destination set $\mathbf D$\\
$u_i$ &  utility function of player $v_i$ & $R_i$ & reward of player $v_i$\\
$\lambda_i$ & cost of player $v_i$ & $a_i$ & action of player $v_i$\\
$\Lambda$ & unit cost for link formation & $U$ & network utility \\
\hline
  \end{tabular}
  \label{tab:notation}
 \end{table*}

\subsection{Network Model and Inter-link Dependency Condition}
\label{subsec:network_model}

We consider a directed graph $\mathcal G$ 
with a set of nodes $\mathcal V(\mathcal G)$ and a set of directed links $\mathcal E(\mathcal
G)$\footnote{
If the considered network changes over time, the network can be modeled 
by a directed graph $\mathcal G_t$ 
with a set of nodes $\mathcal V(\mathcal G_t)$ and a set of directed links $\mathcal
E(\mathcal G_t)$ as a function of time slot $t$. 
However, our focus in this paper
is on the distributed solution for topology
formation at {\emph {each time slot}}, so  
we can omit the subscript $t$ 
in the rest of this paper 
without the loss of generality.}.
An element $v_i \in \mathcal V(\mathcal G)$ can be a source
node and/or a destination node (i.e., data sink)\footnote{
A destination node can be a source node by itself, but, it is not allowed that a
destination node can be a source node for other nodes. 
For example,
a destination node that is connected to 
a server can be a source node by itself (i.e., the server can
have data obtained from the destination node).}, 
and  $x_i$ denotes source data collected by node $v_i$. 
The number of nodes
in $\mathcal V(\mathcal G)$ is denoted by 
$|\mathcal V(\mathcal G)|=N_V$, where
$| \cdot |$ denotes the cardinality of a set. 
Every node in $\mathcal V(\mathcal G)$ plays the role of source by collecting data
(i.e., sensing) and simultaneously plays the role of relay by disseminating the
collected data. 
The set of destination nodes for 
$v_i$ is denoted by $\mathcal D_i  \subseteq \mathcal V(\mathcal G)$, and 
$\mathcal D = \{i | v_i \in \{\mathcal D_1, \ldots, \mathcal D_{N_V}\} \}$ represents 
an index set of destinations for all network nodes. 
The set of destinations is constant and not changed over time.

%
%
A directed link from $v_i$ to $v_j$ is denoted by 
$e_{ij} \in \{0,1 \}$, where 
the {\it active} link ($e_{ij}=1$)  
can deliver data and $X_{ij}$ denotes the data  transmitted via $e_{ij}$.
Otherwise,  
the link is {\it inactive}, and $e_{ij}=0$. 
The link $e_{ij}$ has a direction, i.e., $v_i$ is the tail and $v_j$ is the
head, so that $e_{ij} \neq e_{ji}$. 
In this paper, $e_{ij}$ is called an {\it incoming} link of $v_j$ or an {\it outgoing}
link of $v_i$. 
Note that both $e_{ij}$ and $e_{ji}$ can be simultaneously active, i.e.,
$e_{ij}=1$ and $e_{ji}=1$ \cite{katti2007joint, baik2008network, li2010relay}\footnote{There has been extensive research on the 
full-duplex mode in network coding based relay communication allowing two-way
communications~\cite{katti2007joint, baik2008network, li2010relay}.}. 
$\mathcal E(\mathcal G)$ includes only active links so that 
$|\mathcal E(\mathcal G)|$ is the number of active links in $\mathcal G$.

Let $\delta_{ij}$ be the Euclidean distance between $v_i$ and $v_j$~\footnote{In order to estimate distance between nodes, several works such as \cite{Xu2011node, Sunyong2017} can be adopted in our solution.}. 
As special
cases, we define  $\delta_{ij}=0$ if $v_i = v_j$ and 
$\delta_{ij} = \infty$ if $v_i$ and $v_j$ are not able to make a link between them. 
The set of neighbor nodes of $v_i$ is denoted by $\mathcal H_i = \{ v_j |0< \delta_{ij} \le \Delta\}$, 
where $\Delta$ denotes a connection boundary\footnote{This connection boundary in terms of the Euclidean distance is motivated by realistic communication problems such as limited power budget of a node and interference among wireless nodes.}. 
If $\delta_{ij}>\Delta$, then a link between $v_i$ and $v_j$ cannot be formed, i.e.,
$e_{ij} = e_{ji} = 0$. 
$\tilde
{\mathcal  H}_i^{in} =\{v_j| e_{ji} = 1, \forall v_j \in \mathcal H_i \} $ and 
$ \tilde {\mathcal  H}_i^{out} =\{v_j| e_{ij} = 1, \forall v_j \in \mathcal
H_i \} $ denote a set of 
neighbor nodes of $v_i$ with active incoming and outgoing links, respectively. 
An illustrative example of a sensor network topology with seven nodes 
is shown in Figure~\ref{fig:sensor_NW}. 

For simplicity, we assume that the capacity of  link $e_{ij}$ is one packet per unit
of time slot, i.e., a node can transmit only one packet in each time
slot~\cite{Katti2006, nad2004, topakkaya2011}\footnote{Note that this
assumption can be easily generalized to $\Omega$ bits per unit slot for a constant $\Omega$.}. 
We also assume that each sensor always has a packet to send at each time slot (e.g., a sensor generates a packet for every time slot). Hence, if it builds an outgoing link,  a packet to be transmitted always exists through the outgoing link.  
If a node has multiple outgoing links, it multicasts a single packet per time slot
through all the outgoing links so that all outgoing links from one node deliver the
same packet at the same time slot. If a node has multiple incoming links, it can
receive multiple individual packets by deploying, for example,  multipacket reception
techniques~\cite{Cloud2012, mirrezaei2014}.  
Even though a node can receive multiple packets at a single time slot, under
the conventional SF relaying architecture, a node cannot transmit more
than one packet at a time because of the link capacity constraints. 
Hence, a node becomes a bottleneck of flows when it receives a larger number of
packets than its output link capacity (i.e., one packet per unit of time slot), which
is referred to as the \emph{bottleneck problem}. 

To prevent the bottleneck problem, a node may restrict the number of incoming packets to 
not  exceed the link capacity, and such a constraint can be feasible to restrict the
number of incoming links to at most one, 
%
%
i.e., 
\begin{equation}
  \sum_{\forall v_i \in \tilde{\mathcal H}_j^{in}}e_{ij} \le 1 \,\, \textrm{for all} \,\, v_j \in \mathcal V(\mathcal G)
\label{eqn:inter-link}
\end{equation} 
which is referred to as the {\it inter-link dependency condition} in this paper. 
\begin{figure}[tb]
\centering
\includegraphics[width = 7cm]{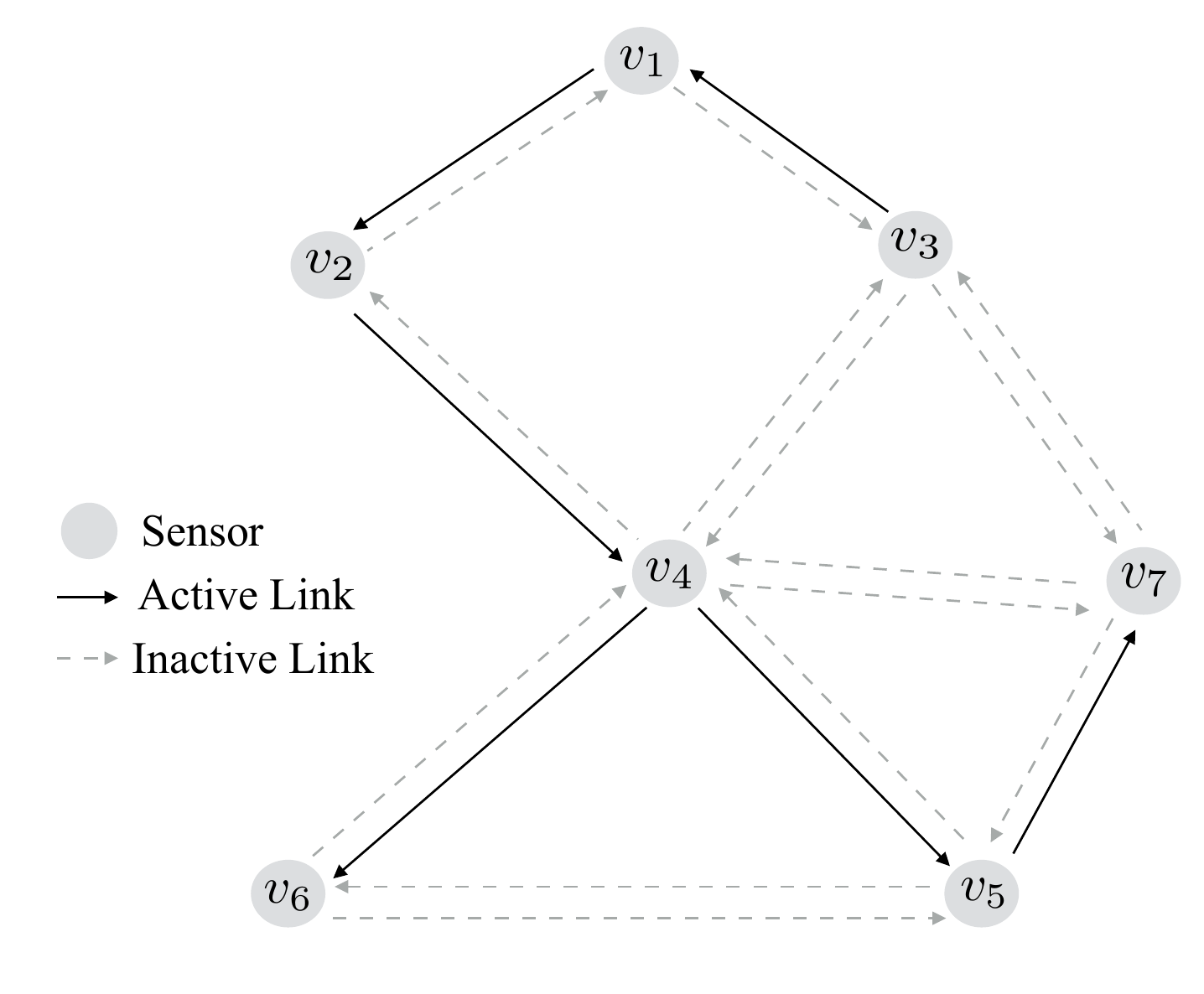}
\caption{ An illustrative example of sensor network topology $\mathcal G$, where $\mathcal
V(\mathcal G) = \{ v_1, \ldots, v_7 \}$ and $\mathcal E(\mathcal G) =
\{e_{12}, e_{24},e_{46},e_{45},e_{57},e_{31}  \}$. 
For node $v_4$, 
a set of neighbor nodes with active incoming and outgoing links are $ \tilde {\mathcal  H}_4^{in}
=\{v_2 \}$ and $\tilde {\mathcal  H}_4^{out} =\{v_5, v_6 \}$, respectively. 
All nodes in this example have at most one incoming link. 
}
\label{fig:sensor_NW}
\end{figure}

\subsection{Elimination of Inter-link Dependency by Network Coding}


The {\it network status} of ${\mathcal G}$ is defined as the set of data 
included in $\mathcal E(\mathcal G)$ with the inter-link dependency condition in \eqref{eqn:inter-link}, which is denoted by 
$J_{\mathcal G}$ and expressed as \eqref{eqn:network_status}. 
\begin{figure*}
\begin{align}
J_{\mathcal G}= \left< \{X_{ij}\cdot e_{ij}| \forall v_i, v_j \in \mathcal V(\mathcal G), i \ne j \}, \left( \sum_{\forall v_i \in \tilde{\mathcal H}_j^{in}}e_{ij} \le 1, \forall v_j \in \mathcal V(\mathcal G) \right) \right>. 
\label{eqn:network_status}
\end{align}
\hrulefill
\end{figure*}
If $e_{ij} \in \mathcal E(\mathcal G)$, then $X_{ij}\cdot
e_{ij} = X_{ij}$. Otherwise, $X_{ij}\cdot e_{ij} = 0$. Hence, $\{X_{ij}\cdot
e_{ij}| \forall v_i, v_j \in \mathcal V(\mathcal G), i \ne j \}$ in
\eqref{eqn:network_status} represents a set of
data included in $\mathcal E(\mathcal G)$ for the link dependent data $X_{ij}$. 



If network coding is deployed in $\mathcal G$, the resulting network status is
denoted by  
$\Phi(J_{\mathcal G})$ and is expressed as 
\begin{align}
\Phi(J_{\mathcal G}  ) = \left< 
 \{p_j\cdot e_{ij}| \forall v_i, v_j \in \mathcal V(\mathcal G), i \ne j \}
\right>
\label{eqn:Phi}
\end{align}
where $p_j$ denotes a network coded packet that flows into $v_j$. 
The network coded packet $p_j= [C_{1j}, \ldots, C_{N_Vj}, y_j]$ is a  
vector of the global coding coefficients $[C_{1j}, \ldots, C_{N_Vj}]^T$ as the header and
$y_j$ as the payload, which is 
constructed as 
\begin{align}
y_j = \sum_{k=1}^{N_V}\bigoplus  \left( C_{kj} \otimes x_k \right)
\label{eqn:y}
\end{align}
where $\oplus$ and $\otimes$ denote the addition and multiplication operations in
a Galois field (GF),
respectively. 
Hence, 
the network coding function $\Phi$ combines all packets that flow into 
$v_j$ and generates a single packet $p_j$.
This operation allows a node to take multiple incoming links and prevents   
the bottleneck problem, so that the inter-link dependency
in~\eqref{eqn:network_status} can be eliminated as in~\eqref{eqn:Phi}.

The elimination of the inter-link dependency 
through the network coding function $\Phi$
can be interpreted as follows. 
The network coding function $\Phi$ converts the link dependent data $X_{ij}$ into 
the link independent
data $p_j$. Hence, 
$v_j$ 
receives a single of packet $p_j$ 
from all incoming links no matter how
many incoming links are formed\footnote{
Note that \eqref{eqn:Phi} does
not mean that the packet $p_j$ is coming to the node $v_j$ for all incoming links
$e_{ij}$. Actual packets in $e_{ij}$ are not all the same as $p_j$. All incoming packets are combined
into $p_j$ based on  
the network coding operation in~\eqref{eqn:y} 
and it can be \emph{interpreted} as the node receives $p_j$ from previous
nodes.  }. 
%
Examples of this interpretation are illustrated in Figure~\ref{fig:NC_function}, and more details of network coding operations are provided in~\ref{sec:NC}.

\begin{figure*}[tb]
\centering
\includegraphics[width = 12cm]{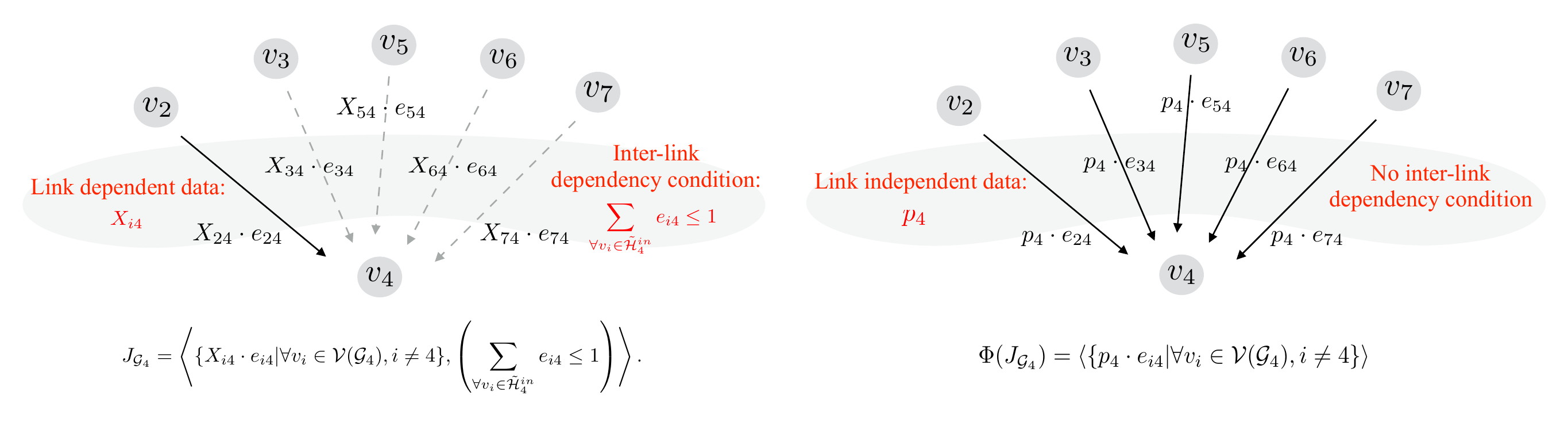}
\caption{
Illustrative examples of the network status of $\mathcal G_4$  without (left) and with (right) 
 network coding function $\Phi$. Note that $\mathcal G_4$ has $v_4$ and its neighbor nodes, which are presented in Figure~\ref{fig:sensor_NW}, i.e.,  $\mathcal G_4 \subset \mathcal G$ and $\mathcal V(\mathcal G_4) = \{v_2, v_3, v_4,v_5, v_6, v_7 \}$. 
(left) Without network coding function $\Phi$ in $v_4$, the inter-link dependency condition must be satisfied and the data are link-dependent. (right) With network coding function $\Phi$ in $v_4$, multiple links can be established simultaneously (i.e., no inter-link dependency condition is considered), and $\Phi$ allows data to be interpreted as link-independent. 
}
\label{fig:NC_function}
\end{figure*}

\section{Distributed Topology Formation Based on Game-theoretic Approaches}
\label{sec:NFG}

In this section, we propose a distributed topology formation strategy 
in a sensor network
with  
multi-source multicast flows. 
We formulate the problem of how to make decisions on 
link connections 
between nodes in the considered network 
as a game, referred to as the network formation game. Then, we show that the network
formation game can be decomposed into link formation games, which enables 
each node to decide which links are active or inactive. Therefore, this 
eventually 
leads to a distributed
solution. 

\subsection{Network Formation Game}
\label{subsec:NFG}

Given 
a set of nodes $\mathcal
V(\mathcal G)$ and 
a destination index matrix 
$\mathbf D = [\mathbf D_1^T, \ldots,\mathbf
D_{N_V}^T]$, where $\mathbf D_i  = [ j | j \in \mathcal D_i ]$ is an index vector
for destination nodes of 
$v_i$,  
a strategic form of the network formation game can be expressed as
\begin{equation}
\notag \mathbf G_{\mathcal G}(\mathbf D) = \langle \mathcal V(\mathcal G),  (\mathbf
a_i)_{v_i \in \mathcal V(\mathcal G)}, ({u_i(a_i, a_{-i},\mathbf D_i)})_{v_i \in \mathcal
V(\mathcal G)} \rangle, 
\end{equation}
where 
$\mathcal V(\mathcal G) $, 
$\mathbf a_{i} = {\bigtimes}_{\forall v_j \in \mathcal H_i} e_{ij} =\bigtimes_{\forall v_j \in
\mathcal H_i}  \{0, 1\}$, and ${u_i(a_i, a_{-i},\mathbf D_i)}$ denote a {finite }set of players, a {finite} set of
actions for player ${v_i}$, and the utility function of player ${v_i}$,
respectively. $\bigtimes$ denotes the Cartesian product. 


A network node $v_i \in \mathcal V(\mathcal
G)$ is a player in the network formation game, which 
makes decisions about link formation with its neighbor 
nodes $\forall v_j \in \mathcal H_i$. The action of $v_i$ is 
denoted by $a_i =(e_{ij})_{\forall v_j \in \mathcal
H_i}\in \mathbf a_i$. 
The utility of $v_i$  is defined as a quasi-linear utility function, expressed as
\begin{align}
u_i(a_i, a_{-i}, \mathbf D_i) &= R_i(a_i, \mathbf D_i) - \lambda_i(a_i, a_{-i})
\end{align}
where $a_{-i}$ denotes a set of actions taken by players other than $v_i$  in $\mathcal V$. 
Given destination nodes, the utility of a player can be determined 
by the reward $R_i(a_i, \mathbf D_i)$ and cost $\lambda_i(a_i, a_{-i})$ associated with its own and others' actions $(a_i, a_{-i})$. More details regarding reward and cost are provided  below.  

{\bf Reward $R_i(a_i, \mathbf D_i)$:} The reward 
represents the distance reduction toward the destination nodes $ \mathbf D_i$ by taking the action $a_i$ at $v_i$, defined as 
\begin{align}
R_i(a_i, \mathbf D_i)= \sum_{v_j \in \mathcal H_i}  E_{ij}(a_i) \left( f\left(\boldsymbol\delta_{j\mathbf D_i}\right) -f\left(\boldsymbol\delta_{i\mathbf D_i}\right) \right)
\label{eqn:reward}
\end{align}
where $E_{ij}(a_i)$ indicates whether  the link $e_{ij}$ is active or
not for an 
action $a_i$, i.e., 
\begin{equation*}
E_{ij}(a_i) = e_{ij} \in \{0, 1\}.
\end{equation*}
For example, if the action $a_i$ makes the link $e_{ij}$ active,  $E_{ij}(a_i) =1$.
Otherwise, 
$E_{ij}(a_i) =0$.  
In \eqref{eqn:reward}, $\boldsymbol\delta_{i\mathbf D_i} = (\delta_{ij})_{\forall j \in \mathbf D_i}$ 
denotes a vector of distances from $v_i$ to destinations $v_j$ for all $j \in \mathbf D_i$, and $f: \mathbb R^{|\mathbf D_i| \times 1}
\rightarrow \mathbb R$ denotes an inversely proportional function such that  $f(\boldsymbol\delta_{i\mathbf D_i})$ is inversely proportional to 
$\boldsymbol\delta_{i\mathbf D_i}$\footnote{
More details about
$f(\boldsymbol\delta_{i\mathbf D_i})$ are discussed in Section~\ref{sec:simulation} with an example (e.g.,
\eqref{eqn:delta_function}-\eqref{eqn:loc}).}. 

The reward function is designed such that a higher reward is given to the node if it 
builds links closer toward  the destination. 
Hence, if $v_j$ is located closer to destination $\mathbf D_i$ than $v_i$,
i.e., $\boldsymbol\delta_{j\mathbf D_i}< \boldsymbol\delta_{i\mathbf D_i}$, then
$f(\boldsymbol\delta_{j\mathbf D_i}) > f(\boldsymbol\delta_{i\mathbf D_i})$, as 
$f$ is an 
inversely proportional function.
This leads to $f\left(\boldsymbol\delta_{j\mathbf D_i}\right)
-f\left(\boldsymbol\delta_{i\mathbf D_i}\right)>0$, which improves 
rewards on $R_i(a_i, \mathbf D_i)$ 
in \eqref{eqn:reward}, if $E_{ij}(a_i)=1$. 
Therefore,  $v_i$ takes the action $a_i$ that builds the link to $v_j$ (i.e.,
$e_{ij}=1$) in order to maximize its own reward. 

The reward function furthermore takes the importance of node locations into
account by assigning higher rewards to closer nodes than the nodes far from the 
destination\footnote{In Figure~\ref{fig:location_impact}, it is confirmed that the
proposed reward function encourages a node close to the destination 
to build a link by
giving higher rewards than a node far from the destination.}.  
Consider two node pairs $(v_i, v_j)$ and $(v_i', v_j')$, for example, 
where both $v_i$
and $v_i'$ have the same destination $\mathbf D$ for simplicity. 
Suppose two node
pairs have
the same distance between them, i.e., 
\begin{equation}
\delta_{i\mathbf D}-\delta_{j\mathbf D} =
\delta_{i'\mathbf D}-\delta_{j'\mathbf D} 
\end{equation}
but the pair $(v_i, v_j)$ is located 
closer to the destination than $(v_i', v_j')$, i.e., 
\begin{equation}
\delta_{i\mathbf D}
<\delta_{i'\mathbf D},  \delta_{j\mathbf D} <\delta_{j'\mathbf D}.
\end{equation} 
Then, 
\begin{align}
f(\delta_{j\mathbf D}) -f\left(\delta_{i\mathbf D}\right) >  f(\delta_{j'\mathbf D}) -f\left(\delta_{i'\mathbf D}\right). 
\end{align}
because of the inversely proportional function $f$. 
Therefore, it is confirmed that 
the definition of the reward function in \eqref{eqn:reward} assigns higher rewards to the
node closer 
to the destination (i.e., $v_i$), 
even though the distance
reduction by making a link is the same. 
Such reward design improves the source-to-destination
connectivity,
because the link formation at a node closer to the destination has critical impact on
network connectivity (i.e., successful connection from sources to destinations).

{\bf Cost $\lambda_i(a_i, a_{-i})$:} 
Given the actions $(a_i, a_{-i})$ selected by players, the cost 
is defined as
\begin{align}
\lambda_i(a_i, a_{-i})=\sum_{v_j \in \mathcal H_i} \left( \frac{E_{ij}(a_i)}{E_{ij}(a_i)+{E_{ji}(a_{j})}} \times \Lambda \right)
\label{eqn:cost} 
\end{align}
where $\Lambda$ is a unit cost for link formation. 
We define ${0}/{0} =0$ \cite{schal1994quadratic,huang1998nonparametric,pele2010quadratic, filipovic2013density}. 
The cost $\lambda_i(a_i, a_{-i})$ in \eqref{eqn:cost}
represents the total payment required for all outgoing links
that $v_i$ makes.
This can be considered as the 
penalty incurred by the
message exchanges and time consumption for negotiation (i.e., required process for link formation between
nodes~\cite{zhang2011game}), or the penalty for causing interference to neighbor nodes\footnote{These are essential for constructing a link between nodes in the
considered network setting as nodes share their services and resources without any
central administration or coordination~\cite{grigoras2007cost}. If the link is bidirectional, i.e.,
both nodes would like to make a connection between them, they can  share 
the cost associated with control message exchanges and the required time consumption for negotiation~\cite{anshelevich2008price}, leading to cost
reduction. However, if the link
is unidirectional, i.e., only one of the nodes would like to solely make a connection to other node, the node is responsible for the cost associated with the message exchanges
and negotiation time.}. 
For a link 
between $v_i$ and $v_j$, 
if either $v_i$ or $v_j$ decides to build the outgoing link, 
the unit cost for link formation $\Lambda$ is solely charged to the node that builds the link. If both
nodes decide to build the link, the link formation cost is charged 
to them equally\footnote{
The equal-division mechanism was first proposed in~\cite{herzog1997} and it has
been extensively deployed in the network formation cost (e.g.,~\cite{  arcaute2009,feigenbaum2001}). 
}. 

The solution to the network formation game $\mathbf G_{\mathcal G}(\mathbf D)$ is 
the set of
actions $(a_i^*, a_{-i}^*)$, which is optimally taken by each player, determining $\mathcal
E(\mathcal G)$ and the corresponding 
network topology. 
While the proposed solution to 
the network formation game $\mathbf
G_{\mathcal G}(\mathbf D)$ can be obtained in a distributed way, the computational
complexity required to find the solution can be significantly increased, especially as
$\mathcal G$ becomes large (i.e., the network size grows). 
Hence, 
we show that the network formation game can be decomposed into
several link formation games 
by deploying network coding, which enables the solution to be found with
significantly lower complexity in the next. 


%
%
%


\subsection{Network Coding Based Game Decomposition}
\label{subsec:NC_game}

We define \emph {edge-disjoint subgraphs} of $\mathcal G$ as a set of subgraphs
whose 
links are disjointed and the union of them is $\mathcal G$\footnote{
There can be
maximum $N_V \choose 2$ edge-disjoint subgraphs in $\mathcal G$ as each link with two
nodes becomes a subgraph of $\mathcal G$. Hence, it is always possible to decompose
$\mathcal G$ into edge-disjoint subgraphs. 
}. 
Specifically, for $N$ edge-disjoint subgraphs $\mathcal L_1, \ldots, \mathcal L_N $
of $\mathcal G$ are 
\begin{itemize}
\item $\mathcal V(\mathcal L_n) \subseteq \mathcal V(\mathcal G)$, 
\item $\mathcal E(\mathcal L_n)
\subseteq \mathcal E(\mathcal G)$, 
\item $\bigcup_{n=1}^{N} \mathcal E(\mathcal L_n) =
\mathcal E(\mathcal G)$ and 
\item $\mathcal E(\mathcal L_n) \cap \mathcal E(\mathcal L_m)
= \emptyset$ for $1\le n, m \le N, n \neq m$.
\end{itemize}
which means that if a graph is decomposed into multiple edge-disjoint subgraphs, the vertices are allowed to be shared across subgraphs, but edges are not~\cite{Agarwal2009}\footnote{By the definition of edge-disjoint subgraph, a vertex can be shared in multiple
subgraphs such that data can be conveyed from one subgraph to another subgraph via the 
shared vertex. This allows data packets in the source nodes to be delivered to 
the destination nodes.}. 
The network formation game for a subgraph $\mathcal L_n$ with $\mathbf D$ can be expressed as 
\begin{equation*}
 \mathbf G_{\mathcal L_n}(\mathbf D)= \langle \mathcal V(\mathcal L_n),  (\mathbf a_i)_{v_i
\in \mathcal V(\mathcal L_n)}, (u_i(a_i, a_{-i}, \mathbf D_i))_{v_i \in \mathcal V(\mathcal L_n)}
\rangle
\end{equation*}
and the network
status for the resulting network from $\mathbf G_{\mathcal L_n}(\mathbf D)$ is
denoted by $J_{\mathbf G_{\mathcal L_n}(\mathbf D)}$, as defined in
\eqref{eqn:network_status}.

Since the actions simultaneously determined by the players in a game are the union of
the links that are active and inactive in the network, the 
product operation for games can be considered as 
the union of their network status, expressed as 
\begin{align}
\prod_{n=1}^{N} \mathbf G_{\mathcal L_n}(\mathbf D) \triangleq \bigcup_{n=1}^{N}J_{\mathbf
G_{\mathcal L_n}(\mathbf D)}.
\label{eqn:sum}
\end{align}
In Theorem~\ref{lem:indep}, we show that network coding can decompose 
the network formation game 
$\mathbf G_{\mathcal G}(\mathbf D)$ into independent games $\mathbf G_{\mathcal
L_n}(\mathbf D)$ for subgraph  $\mathcal L_n$
 for $1\le n \le N$.


\begin{theorem} 
The network formation game for a graph can be decomposed by network coding into
independent games for edge-disjoint subgraphs. 
\label{lem:indep}
\end{theorem}

\begin{proof}
To show that  
the network formation game for a graph can be decomposed by network coding into
independent games for edge-disjoint subgraphs, 
it should be proved that 
\begin{align}
\Phi (  {\mathbf G_{\mathcal G}}(\mathbf D) )
=\prod_{n=1}^{N }\Phi( {\mathbf G_{\mathcal L_n}} (\mathbf D))
\label{eqn:lem_prf}
\end{align}
where  $\Phi$ is the network coding function defined in \eqref{eqn:Phi}.

The network formation game for a graph $\mathcal G$ is the joint game of
edge-disjoint subgraphs ${\mathcal L_1, \ldots, \mathcal L_N}$, which can be played
as sequential conditional games based on a chain rule as in \eqref{eqn:th-game}--\eqref{eqn:lem_proof_1}. 
\begin{figure*}[bt]
\begin{align}
\mathbf G_{\mathcal G} (\mathbf D) &=  {\mathbf G_{\mathcal L_1, \mathcal L_2, \cdots, \mathcal L_N }} (\mathbf D)\notag\\
&= {\mathbf G_{\mathcal L_1}} (\mathbf D) \cdot {\mathbf G_{\mathcal L_2 | \mathcal L_1 }} (\mathbf D) \cdots {\mathbf G_{\mathcal L_N | \mathcal L_1, \mathcal L_2, \cdots, \mathcal L_{N-1} }} (\mathbf D)\label{eqn:th-game}\\
&= J_{\mathbf G_{\mathcal L_1}(\mathbf D)} \cup J_{\mathbf G_{\mathcal L_2 | \mathcal L_1 } (\mathbf D)} \cup \cdots \cup J_{\mathbf G_{\mathcal L_N | \mathcal L_1, \mathcal L_2, \cdots, \mathcal L_{N-1} } (\mathbf D)}\label{eqn:th-status}\\
&= \left< \{X_{ij}\cdot e_{ij}| \forall v_i, v_j \in \mathcal V(\mathcal L_1), i \ne j \}, \left( \sum_{\forall v_i \in \tilde{\mathcal H}_j^{in}}e_{ij} \le 1, \forall v_j \in \mathcal V(\mathcal L_1) \right) \right>\notag\\
& \bigcup \left< \{X_{ij}\cdot e_{ij}| \forall v_i, v_j \in \mathcal V(\mathcal L_2), i \ne j \}, \left( \sum_{\forall v_i \in \tilde{\mathcal H}_j^{in}}e_{ij} \le 1, \forall v_j \in \mathcal V(\mathcal L_1)\cup  \mathcal V(\mathcal L_2) \right) \right>\notag\\
& \bigcup \cdots \bigcup
\left< \{X_{ij}\cdot e_{ij}| \forall v_i, v_j \in \mathcal V(\mathcal L_N), i \ne j \}, \left( \sum_{\forall v_i \in \tilde{\mathcal H}_j^{in}}e_{ij} \le 1, \forall v_j \in \bigcup_{n=1}^{N}\mathcal V(\mathcal L_n) \right) \right>
\label{eqn:lem_proof_1}
\end{align}
\hrulefill
\end{figure*}
Here, the equality between \eqref{eqn:th-game} and \eqref{eqn:th-status} is based
on \eqref{eqn:sum}, and \eqref{eqn:lem_proof_1} is based on the definition of the network
status in \eqref{eqn:network_status}. 
Note that the network formation game expressed in 
\eqref{eqn:lem_proof_1} still includes the inter-link dependency.

By applying the network coding function $\Phi$ in 
\eqref{eqn:lem_proof_1}, we have 
\begin{align}
\Phi (&\mathbf G_{\mathcal G} (\mathbf D) ) \\
&= \left<  \{p_j\cdot e_{ij}| \forall v_i, v_j \in \mathcal V(\mathcal L_1), i \ne j \} \right>\notag\\
&\quad \cup \left<  \{p_j\cdot e_{ij}| \forall v_i, v_j \in \mathcal V(\mathcal L_2), i \ne j \} \right>\notag\\ 
&\quad \cup \cdots \cup 
\left<  \{p_j\cdot e_{ij}| \forall v_i, v_j \in \mathcal V(\mathcal L_N), i \ne j \} \right>\label{eqn:nc-status}\\
&= \Phi(J_{\mathcal L_1(\mathbf D)}) \cup  \Phi(J_{\mathcal L_2(\mathbf D)})\cup \cdots \cup \Phi(J_{\mathcal L_N(\mathbf D)})\label{eqn:nc-J}\\
&= \Phi({\mathbf G_{\mathcal L_1}} (\mathbf D))\cdot\Phi({\mathbf G_{\mathcal L_2}} (\mathbf D))\cdots  \Phi({\mathbf G_{\mathcal L_N}} (\mathbf D))\label{eqn:nc-times}\\
&=\prod_{n=1}^{N }\Phi( {\mathbf G_{\mathcal L_n}} (\mathbf D)) \notag
\end{align}
where the equality between \eqref{eqn:nc-status} and \eqref{eqn:nc-J} is based on \eqref{eqn:Phi}, and \eqref{eqn:nc-times} is based on \eqref{eqn:sum}. Therefore, the network formation game for a graph can be decomposed by network coding into
independent games for edge-disjoint subgraphs, which completes the proof. 
\end{proof}

Importantly, Theorem~\ref{lem:indep} implies that $\mathbf G_{\mathcal G}(\mathbf D)$
with multiple destinations in $\mathbf D$ 
can be further decomposed into 
$\mathbf G_{\mathcal L_n} (d)$ for $1 \le n \le N$ and $d \in \mathcal
D$ with a single destination node $v_d$, which is shown in Theorem~\ref{th:unicast}.

In order to prove this, 
we define 
a virtual subnode of $v_i$ that has flows to be delivered to 
destination $v_d$ as 
$v_i(d)$. 
By definition, there are $|\mathcal D|$ virtual subnodes 
in $v_i$. Similarly, a virtual sublink of $e_{ij}$ with destination $v_d$ is 
denoted by $e_{ij}(d)$, 
and $e_{ij}$ includes $|\mathcal D|$ sublinks. 
Then,  a virtual subgraph for destination $v_d$ can be defined as $\mathcal L_n (d)$,
which satisfies
\begin{itemize}
  \item $\mathcal V(\mathcal L_n (d)) \subseteq \mathcal V(\mathcal L_n
(\mathbf D))$, 
\item $\mathcal E(\mathcal L_n (d)) \subseteq \mathcal E(\mathcal L_n
  (\mathbf D))$, 
\item  $\bigcup_{d \in \mathcal D} \mathcal E(\mathcal L_n (d)) = \mathcal
E(\mathcal L_n (\mathbf D))$, and
\item $\mathcal E(\mathcal L_n (d)) \cap \mathcal
E(\mathcal L_n (d')) = \emptyset$ for $d, d' \in \mathcal D, d \neq d'$.
\end{itemize}


\begin{theorem} 
The network formation game with multicast flows can be decomposed by network coding 
into independent games with unicast flows for edge-disjoint subgraphs.
\label{th:unicast}
\end{theorem}

\begin{proof}

In this proof, we show that $\Phi ( {\mathbf G_{\mathcal G}}(\mathbf D))$ can be
decomposed into $\Phi( {\mathbf G_{\mathcal L_n}} ( d))$ for $1 \le n \le N$ and $d
\in \mathcal D$. 

In Theorem~\ref{lem:indep}, it is shown that 
\begin{align}
\Phi (\mathbf G_{\mathcal G}(\mathbf D)) 
&= \prod_{n=1}^{N}\Phi \left(\mathbf G_{\mathcal L_n}(\mathbf D) \right).
\label{eqn:th_proof_def}
\end{align}
Since a subgraph $\mathcal L_n(\mathbf D)$ can be decomposed into virtual subgraphs
$\mathcal L_n( d), \forall d \in \mathcal D$,   
the game $\mathbf G_{\mathcal L_n}(\mathbf
D)$ with network coding can also be decomposed 
into independent games 
$\mathbf
G_{\mathcal L_n}(d), \forall d \in \mathcal D$
for virtual subgraphs based on Theorem~\ref{lem:indep}, i.e., 
\begin{equation}
\Phi \left(\mathbf G_{\mathcal L_n}(\mathbf D) \right) = \prod_{d \in \mathcal D} \Phi \left( \mathbf G_{\mathcal L_n} (d)\right). 
\label{eqn:lem_proof_2} 
\end{equation}
Therefore, we can conclude from \eqref{eqn:th_proof_def} and \eqref{eqn:lem_proof_2} 
that
\begin{align}
\Phi (\mathbf G_{\mathcal G}(\mathbf D)) 
=  \prod_{n=1}^{N}\prod_{d \in \mathcal D} \Phi \left( \mathbf G_{\mathcal L_n} (d)\right), 
\end{align}
which completes the proof. 
\end{proof}
Theorem~\ref{th:unicast} implies that  network coding allows the network formation
game for multicast flows (i.e., $\mathbf G_{\mathcal L_n}(\mathbf D)$) to be
decomposed into independent games with unicast flows for edge-disjoint subgraphs
(i.e., $\mathbf G_{\mathcal L_n} (d), \forall d \in \mathcal D$). 
Moreover, Theorem~\ref{th:unicast} enables  
the topology of a network with  
multi-source multicast flows to be determined in a distributed way, 
by solving independent games of edge-disjoint subgraphs with unicast flows, referred
to as the 
\emph{link formation game} in this paper. More details about the link formation game
are given in the next section. 
An illustrative example of the network formation and link formation games are shown in Figure~\ref{fig:game_decomposition}.
\begin{figure}[tb]
\centering
\includegraphics[width = 8cm]{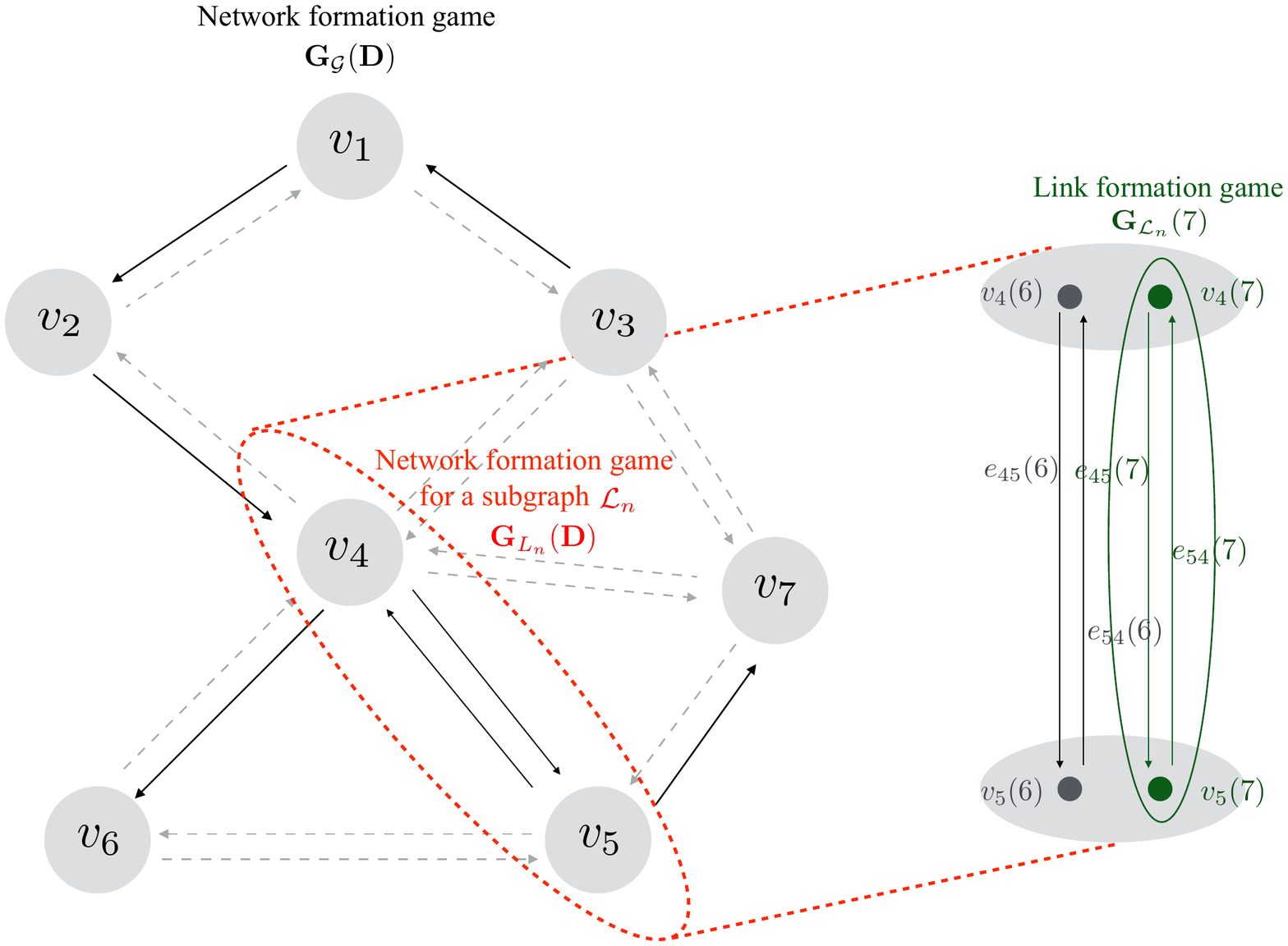}
\caption{The decomposition of the network formation game into link formation
games for the network illustrated in Figure~\ref{fig:sensor_NW}. 
The network formation game $\mathbf G_{\mathcal G}(\mathbf D)$ is decomposed into
independent games of edge-disjoint subgraphs  $\mathbf G_{\mathcal L_n}(\mathbf D)$, 
which are further decomposed into link formation games, $\mathbf G_{\mathcal L_n}(d)$. Note
that the decomposition can be permitted only by deploying network coding. 
}
\label{fig:game_decomposition}
\end{figure}

\subsection{Link Formation Games and Distributed Topology Design}
\label{sec:LFG}

As discussed in Section~\ref{subsec:NC_game}, a link formation game consists of two
players with a unicast flow. 
%
%
The strategic form of the link formation game can be expressed as
\begin{equation}
\mathbf G_{\mathcal L}(d)= \langle \mathcal V(\mathcal L),  (\mathbf a_i)_{v_i \in
\mathcal V(\mathcal L)}, (u_i(a_i, a_{-i},d))_{v_i \in \mathcal V(\mathcal L)} \rangle 
\label{eqn:lfg_def}
\end{equation}
where $\mathcal V(\mathcal L) = \{v_i, v_j\}$ and $\mathbf a_i = \{0, 1\}$ denote 
a player set and an action set 
for player ${v_i}$ for destination $v_d$, respectively. 
The utility function is expressed as 
\begin{align}
u_i(a_i, a_{j}, d) &= R_i(a_i, d) - \lambda_i(a_i, a_j)\\
&= a_i \left( f(\delta_{jd}) -f\left(\delta_{id}\right) \right) -  \Lambda \cdot \frac{a_i}{a_i+a_{j}} 
\label{eqn:LFG_util_2}
\end{align}
where $f(\delta_{id}):\mathbb R \rightarrow \mathbb R$ is an
inversely proportional function 
of $\delta_{id}$. 
For the link formation game, the cost function can be expressed as
\begin{align}
\notag \lambda_i(a_i, a_{j})= \Lambda \cdot \frac{a_i}{a_i+a_{j}}  =
    \begin{cases}
      \Lambda  , & \text{if}\ a_i=1 ,a_{j}=0  \\
      \Lambda/2, &\text{if}\ a_i=1 ,a_{j}=1  \\
      0, & \text{if}\ a_i=0
    \end{cases}.   
\end{align}
The corresponding normal form of the link formation game 
is shown in Table~\ref{tab:normal_form}. 
\begin{table*}[tb]
  \caption{The normal form of the link formation game}
\centering
  \begin{tabular}{ |c || c | c|  }
    \hline
   ($u_{i}, u_{j},d$) & $a_j = 1$ & $a_j  = 0$ \\ \hline\hline
    $a_i = 1$ & ${(R_i(1, d)}-\frac{\Lambda}{2}, {R_j(1, d)}-\frac{\Lambda}{2})$ &  ${(R_i(1, d)}-\Lambda, 0)$ \\
 \hline
    $a_i = 0$ & $(0, {R_j(1, d)}-\Lambda)$ & $ (0,0)$ \\
  \hline
  \end{tabular}
  \label{tab:normal_form}
 \end{table*}


As a solution concept for the link formation game,
we adopt the pure strategy Nash equilibrium (NE). A pure strategy NE $(a_i^*, a_j^*)$ 
for $v_i$ and $v_j$ 
can be 
expressed as 
\begin{equation}
  u_i (a_i^*, a_{j}^*,d) \ge u_i (a_i, a_{j}^*,d) \,\,\textrm{for all} \,\, a_i \in
  \mathbf a_i
\label{eqn:NE}
\end{equation}
and 
\begin{equation}
  u_j (a_i^*, a_{j}^*,d) \ge u_j (a_i^*, a_{j},d) \,\,\textrm{for all} \,\,
  a_j \in
  \mathbf a_j.
\label{eqn:NEj}
\end{equation}
If multiple pure strategy NEs exist, the set of pure strategy NEs is denoted by 
\begin{align}
 \notag \mathbf A^* = \{(a_i^*, a_{j}^*)|& u_i (a_i^*, a_{j}^*,d) \ge
 u_i (a_i, a_{j}^*,d),\forall a_i \in \mathbf a_i,  \\
 \notag & u_j (a_i^*, a_{j}^*,d) \ge
u_j (a_i^*, a_{j},d),  \forall a_j \in \mathbf a_j \}.
\end{align}
The pure strategy NE enables
nodes $v_i$ and $v_j$ to decide which  
outgoing links are
active or inactive, 
resulting in 
a stable network topology  
$\mathcal E(\mathcal L)$. 

\floatname{algorithm}{Algorithm}
\algsetup{indent= 1em}
\begin{algorithm}[tb]
        \caption{Algorithm for Distributed Topology Design Based on Link Formation
		Games}
        \label{alg:bigpicture}
\begin{algorithmic}[1]
\smallskip
\REQUIRE a set of nodes $\mathcal V(\mathcal G)$, 
sets of neighbor nodes $\mathcal H_i$ for $v_i \in \mathcal V(\mathcal G)$, 
an index set of destinations $\mathcal D$, 
utility function $u_i(a_i,a_{j}, d)$ for $v_i \in \mathcal V(\mathcal G)$
\vspace{0.2cm}
\STATE \text{\bf Initialize:} {$\mathcal E(\mathcal G) = \emptyset$}\vspace{0.1cm}
\STATE Decompose $\mathcal V(\mathcal G)$ into a set of node pairs $\mathcal V(\mathcal L_n)$ for $1 \le n \le {N_V \choose 2}$
\vspace{0.1cm}
\FOR {$n$ = 1: $N_V \choose 2$ }
\STATE $(v_i, v_j) \leftarrow \mathcal V(\mathcal L_n)$ \texttt{//assign players}
\vspace{0.2cm}
\FOR {$d \in \mathcal D$}
\STATE \text{\bf Initialize:} $(a_i^*, a_j^*) \leftarrow (0, 0)$, $flag =0$
\vspace{0.2cm}
\STATE \texttt{ // begin link formation game} \vspace{0.2cm}
\WHILE {$flag =0$}\vspace{0.1cm}
\STATE $temp\_a \leftarrow (a_i^*, a_j^*)$\vspace{0.1cm}
\STATE $a_i^* \leftarrow \arg\max_{a_i} u_i(a_i, a_{j}^*,d)$\vspace{0.1cm}
\STATE $a_j^* \leftarrow  \arg\max_{a_j} u_j(a_i^*, a_{j},d)$\vspace{0.2cm}
\IF {$temp\_a = (a_i^*, a_j^*)$}
\STATE $flag \leftarrow 1$ 
\ENDIF
\ENDWHILE
\STATE $e_{ij} \leftarrow a_i^*$, $e_{ji} \leftarrow a_j^*$\vspace{0.2cm}
\ENDFOR
\ENDFOR
\RETURN $\mathcal E(\mathcal G)$
\end{algorithmic}
\end{algorithm}
The steps for the proposed solution are described in
Algorithm~\ref{alg:bigpicture}.
In Algorithm~\ref{alg:bigpicture},  nodes $\mathcal V(\mathcal G)$ are decomposed 
into $N = {N_V \choose 2}$ sets of node pairs $\mathcal V(\mathcal L_n)$ 
for $1 \le n \le N$. 
Then, the link formation game is formulated as 
$\mathbf G_{\mathcal L_n}(d)= \langle \mathcal V(\mathcal L_n),
(\mathbf a_i)_{v_i \in \mathcal V(\mathcal L_n)}, (u_i(a_i, a_{-i},d))_{v_i \in \mathcal
V(\mathcal L_n)} \rangle$
given $\mathcal V(\mathcal L_n)$ and a destination node $v_d$. 
The link formation games, 
$\mathbf G_{\mathcal L_n}(d)$ for $1 \le n \le N$, $\forall d \in
\mathcal D$, 
are solved by finding 
a pure strategy NE, and 
all the active links can eventually be included in $\mathcal E(\mathcal G)$.

Note that Algorithm~\ref{alg:bigpicture} is guaranteed to determine at least one
stable topology, as shown in Theorem~\ref{th:NE_proof}.


\begin{theorem}
It is guaranteed that at least one topology can be determined by Algorithm~\ref{alg:bigpicture}. 
\label{th:NE_proof}
\end{theorem}

\begin{proof}
See \ref{sec:NE_proof}. 
\end{proof}

Based on Theorem~\ref{th:NE_proof}, we confirm that Algorithm~\ref{alg:bigpicture}
can always provide at least one network topology for the multi-source multicast network. In the next
section, we provide performance evaluation based on simulation results.  

\section{Simulation Results}
\label{sec:simulation}

In this simulation, we consider a sensor network
in which all sensors collect data, and some of them are destination nodes. 
Sensor nodes aim to deliver their collected data to their destination nodes in 
$\mathcal D_i$ by making links between them. 


\subsection{Simulation Setup}
We consider $N_V$ nodes in a cell with a radius of $R$, 
where the nodes are randomly located over the cell 
based on a uniform distribution. 
Since a node $v_i$ would like to deliver its collected data to its destination nodes
in $\mathcal D_i$,  
%
there are $\sum_{i = 1}^{N_V} |\mathcal D_i|$ total flows in the network.  
The connection boundary is set as $\Delta = R$ so that 
the neighbor nodes of $v_i$ are determined as $\mathcal H_i = \{ v_j| 0 < \delta_{ij} \le R\}$. 
Note that the longest distance between two edge nodes of the cell can be $2R$ such that our simulation  setting $\Delta = R$ may induce a multi-hop network. 
Each node makes its own decisions for outgoing link formation with neighbor nodes
based on  
Algorithm~\ref{alg:bigpicture}. 

In the simulations, we use a function 
$f(\delta_{id})$ for the utility function 
in \eqref{eqn:LFG_util_2},
defined as
\begin{equation}
f(\delta_{id}) =  \frac{1}{\delta_{id}^2 + 1}
\label{eqn:delta_function} 
\end{equation}
which satisfies the requirements for $f(\delta_{id})$
discussed in Section~\ref{subsec:NFG}, 
i.e.,   
it is
inversely proportional to $\delta_{id}$ and $f(\delta_{jd})
-f(\delta_{id}) > 0$ for $\delta_{id} > \delta_{jd}$. 
Hence, 
if $v_i$ makes an outgoing link to $v_j$, which is closer to a destination node
than $v_i$, then positive rewards are given to $v_i$. 
For example, consider two node pairs $(v_i, v_j)$ and $(v_i', v_j')$ that have 1)
the same distance between them, i.e., $\delta_{id}-\delta_{jd} =
\delta_{i'd}-\delta_{j'd}$, but 2) different locations, i.e., $\delta_{id}
<\delta_{i'd},  \delta_{jd} <\delta_{j'd}$. Then, 
\begin{align}
f(\delta_{jd}) -f\left(\delta_{id}\right) >  f(\delta_{j'd}) -f\left(\delta_{i'd}\right), 
\label{eqn:loc}
\end{align}
which implies that the nodes close to the destination $(v_i, v_j)$ receive more 
reward than the distant nodes  $(v_i', v_j')$. 
The utility function of the link formation game in \eqref{eqn:LFG_util_2} can
be correspondingly expressed as
\begin{align}
u_i(a_i, a_{j}, d) 
&= a_i \left( \frac{1}{\delta_{jd}^2 + 1} -\frac{1}{\delta_{id}^2 + 1} \right) -
\left( \frac{a_i}{a_i+a_{j}} \cdot \Lambda \right)\label{eqn:lfg_simul}.
\end{align}

We finally define {\it network utility} as a measure of  
 the resulting networks performance, expressed as  
\begin{align*}
U(a_i, a_{-i}, \mathbf D) = \sum_{\forall v_i \in \mathcal V(\mathcal G)} \left( \sum_{d\in\mathbf D} R_i(a_i, d) - \lambda_i (a_i, a_{-i})
  \right)
\end{align*}  
which includes both the total rewards from 
all the destination nodes and the costs required to make links in the networks.
Therefore, the network utility can be used  to quantify 
how many rewards can be earned by the nodes while reducing the costs of 
link formation. 

\subsection{Numerical Analysis of the Proposed Topology Design}

In this section, we numerically analyze several aspects of the proposed
algorithm for the topology formation implemented by 
Algorithm~\ref{alg:bigpicture}. 
In the simulations, we consider two destination nodes, i.e., $|\mathcal D_i|=2$ for
$1 \le i \le N_V$, which are randomly determined in each experiment unless otherwise stated. 
The connection boundary is set as $R = 10$ 
and all the experiment results are averaged from $1,000$ independent experiments.

Figure~\ref{fig:location_impact} shows the effects of node locations on the
probability of link formation. 
\begin{figure}[tb]
  \centering
\includegraphics[width = 9cm]{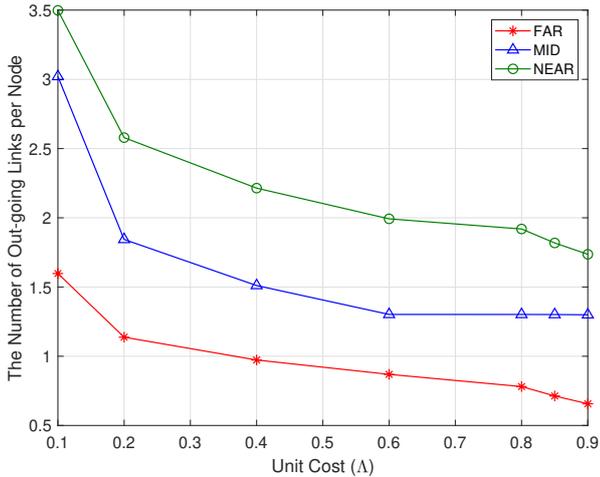}
\caption{The number of outgoing links per node for $N_V = 50$. 
}
\label{fig:location_impact}
\end{figure}
In the experiment, 
we consider two adjacent nodes 
as the destination for all nodes, located at the cell edge,
i.e.,  $\mathcal D_1 = \cdots = \mathcal D_{N_V}$. 
The nodes in the network are classified as three types based on their distance from the 
destination nodes: those that are close to the destinations (NEAR), far from the
destinations (FAR), and in the middle of them (MID).  
Figure~\ref{fig:location_impact} clearly confirms that 
the nodes closer to the destinations make more outgoing links. This is because 
$f(\delta_{id})$ in the utility function 
enables nodes closer to the destinations 
to obtain more rewards by making outgoing links. 

\begin{figure}[tb]
\centering
\includegraphics[width = 9cm]{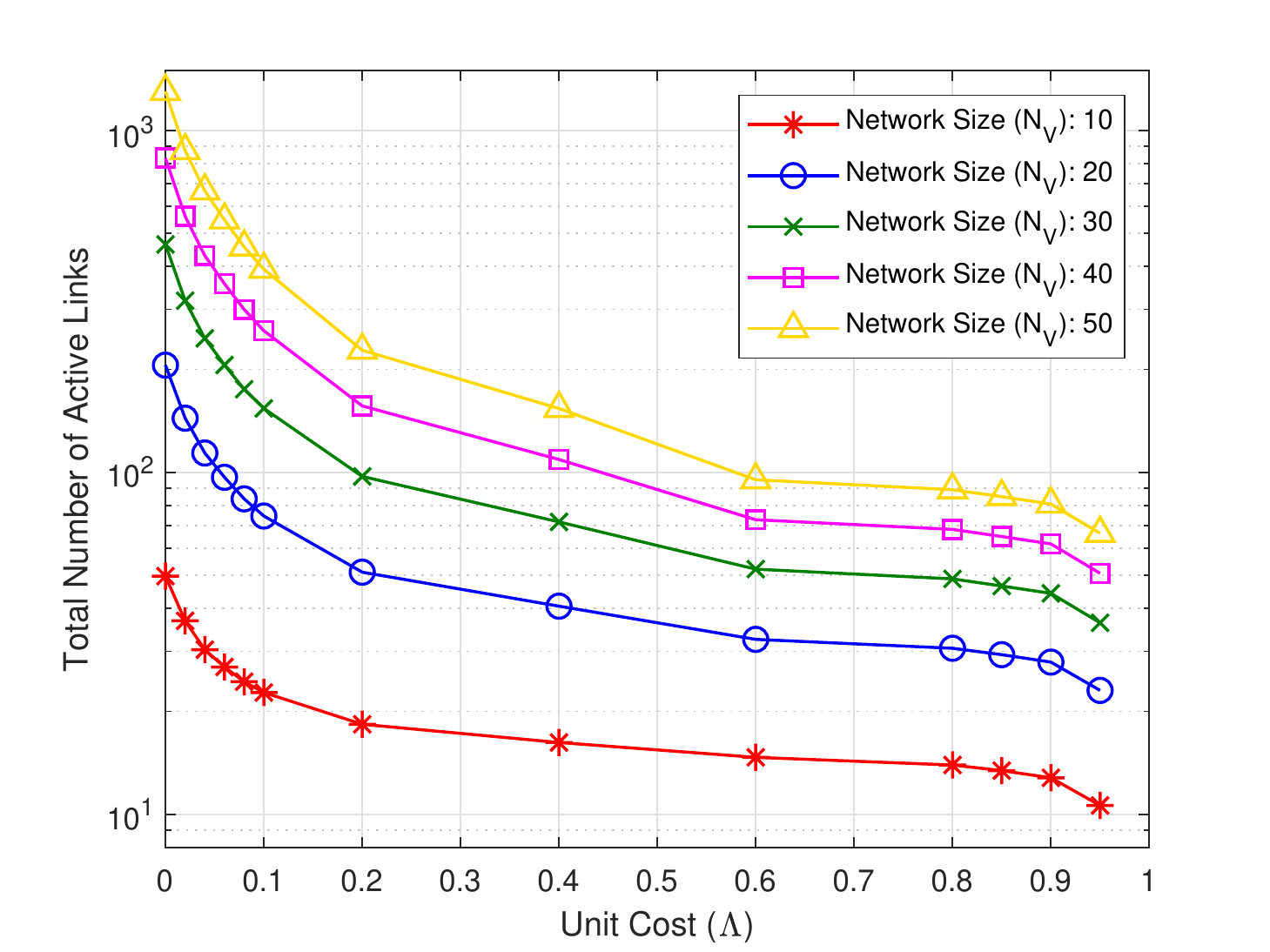}
\caption{ The number of active links for unit costs in various network sizes. }
%
\label{fig:num_of_links}
\end{figure}
Figure~\ref{fig:num_of_links} shows the total number of active links in a network
($|\mathcal E(\mathcal V)|$) for various unit costs ($\Lambda$) and 
network sizes 
($N_V$). 
It can be confirmed that a smaller number of 
active links is included in a resulting network topology 
as the network size
decreases or 
the unit cost increases. 
This is because the nodes with more neighbor nodes or with lower unit cost
can make a larger number of
active links.

The number of active links in a network topology can influence 
the successful
connection from source nodes to destination nodes. 
In order to evaluate
the impact of the number of active links on successful network formation, we
define the {\it connection failure ratio} as the 
number of
disconnected flows over the total number of flows ($\sum_{i=1}^{N_V} |\mathcal
D_i|$). 
In Figure~\ref{fig:plr}, connection failure ratios for unit costs are presented. 
\begin{figure}[tb]
\centering
\includegraphics[width =9cm]{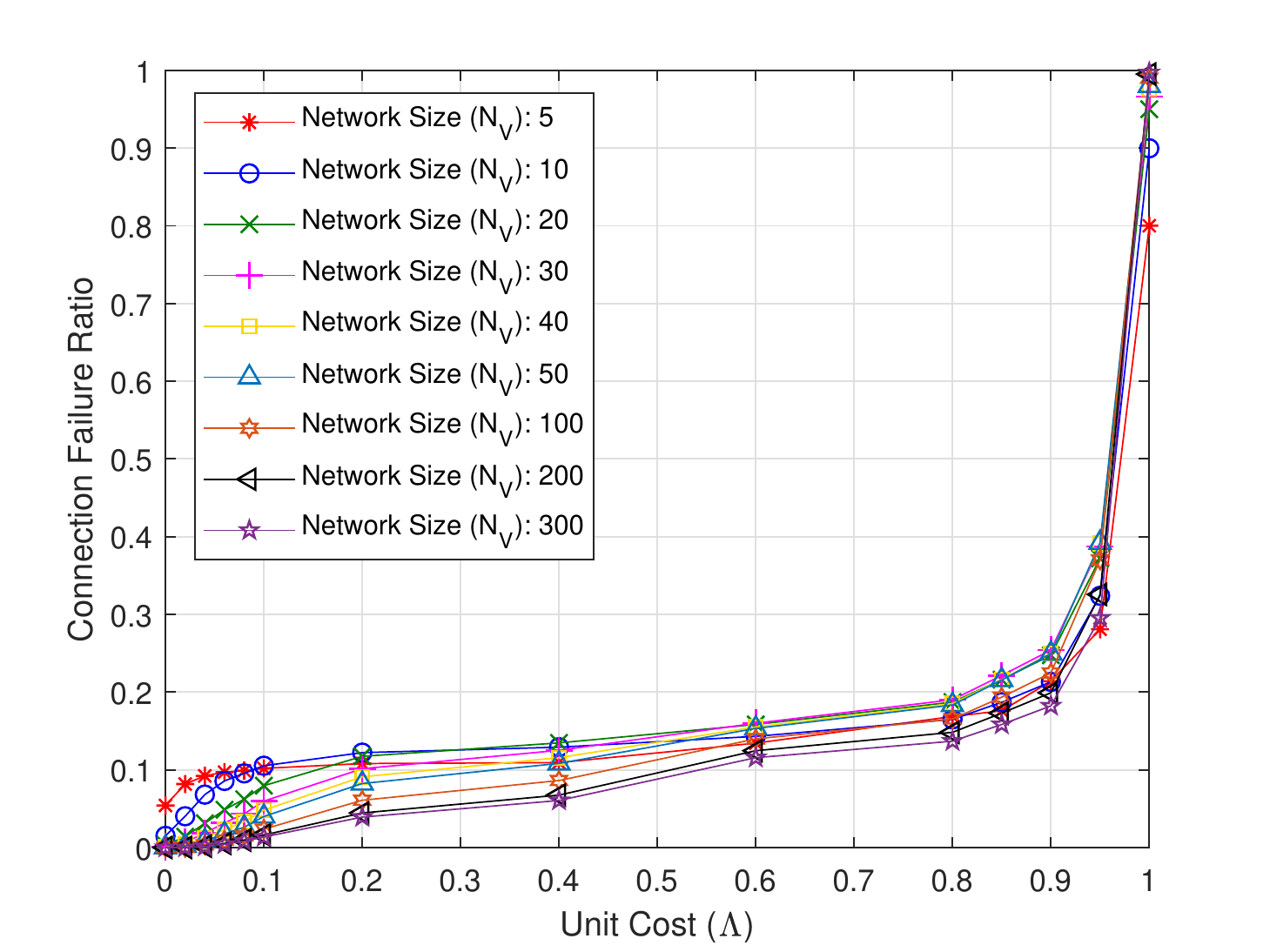}
\caption{ Connection failure ratios for unit costs in various network sizes. 
}
\label{fig:plr}
\end{figure}

As shown in Figure~\ref{fig:num_of_links},
the number of active links
decreases rapidly as $\Lambda$ increases
in the range of  small 
$\Lambda$ values (e.g., $ 0 \le \Lambda \le 0.2$). However,   
this does not significantly affect
the connection failure ratios as shown  in Figure~\ref{fig:plr}.
In contrast, 
a small number of links
can be active in the range
of  large $\Lambda$ values (e.g., 
$0.8 \le \Lambda  \le 1$), 
which significantly 
increases the connection failure ratio, thereby resulting in a high probability of unsuccessful
data delivery.  
Hence, 
a network can be sustainable in terms of successful
data delivery only if the number of active links is large enough to take advantage of path diversity. 
Note that 
Algorithm~\ref{alg:bigpicture} is scalable to the network size
in terms of the connection failure ratio (or
network topology formation) since the impact of network size on the connection
failure ratios is limited, as shown in Figure~\ref{fig:plr}. Therefore, 
the proposed algorithm can be deployed in large-scale networks.  

In Figure~\ref{fig:Net_utility_NW_size}, the network utilities for various network sizes 
and unit costs are shown. 
\begin{figure}[tb]
\centering
\includegraphics[width = 9cm]{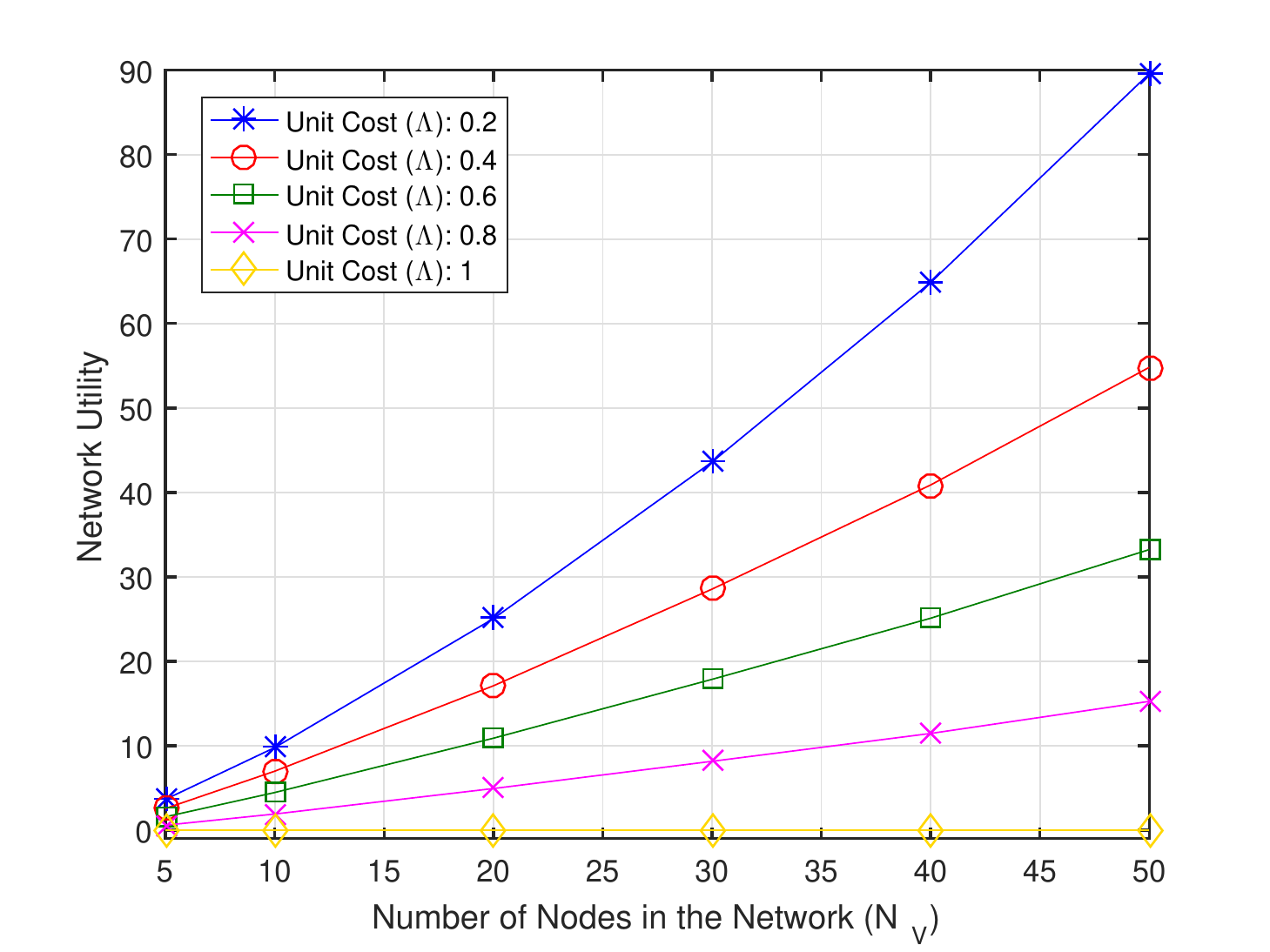}
\caption{ Network utilities for network sizes with various unit costs. }
\label{fig:Net_utility_NW_size}
\end{figure}
The network utility increases as $N_V$ increases (i.e., more rewards can be achieved) 
or $\Lambda$
decreases (i.e., the cost for link formation is lowered). 
In the experiments, for example, 
network nodes decide not to build any links 
if $\Lambda =1$, 
achieving no network utility.

\begin{figure}[tb]
\centering
\includegraphics[width = 9cm]{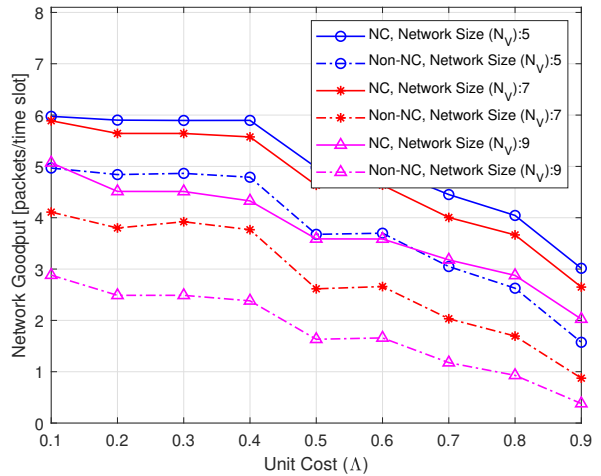}
\caption{ 
Network goodput with and without network coding for unit costs in various network sizes. }
\label{fig:goodput}
\end{figure}
\begin{figure}[tb]
\centering
\includegraphics[width = 9cm]{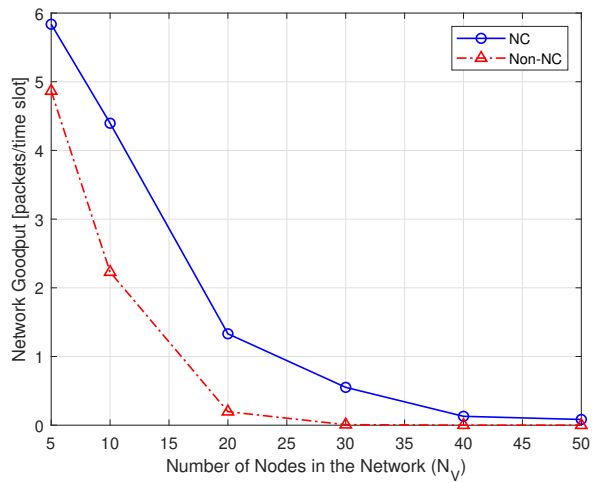}
\caption{ 
Network goodput with and without network coding over network sizes when $\Lambda = 0.1$. }
\label{fig:goodput_NWsize}
\end{figure}
%
We next evaluate the performance of the  
proposed algorithm in terms of 
the network goodput, measured by 
the number of packets successfully delivered to the destinations 
per time slot. 
The results are shown in Figure~\ref{fig:goodput} and Figure~\ref{fig:goodput_NWsize}.
In this experiment, the connection boundary is set as $\Delta = 1.1R$, and all nodes in the network 
generate and deliver packets toward two destination nodes. 
They make decisions on
link formation 
based on 
Algorithm~\ref{alg:bigpicture}.
It is obvious that network goodput decreases as unit cost increases because the 
connection
failure ratio also increases, 
as confirmed in Figure~\ref{fig:plr}.
Figure~\ref{fig:goodput} shows that higher network goodput can be achieved by deploying
network coding, which  
means that the proposed algorithm successfully builds the
network while taking advantage of network coding. 
In Figure~\ref{fig:goodput} and Figure~\ref{fig:goodput_NWsize}, it is observed that network goodput decreases as
network size increases. 
This is because 
the number of hops required for a packet to arrive at the
destination increases as the network size is enlarged, taking a 
longer time 
for packet delivery. 
The results are well-aligned with  \cite{gupta2000}, which theoretically proves that the goodput scales as $\mathcal O(1/\sqrt{N_V \log N_V})$ in random networks and $\mathcal O(1/\sqrt{N_V})$ in the optimal networks. 

%



In the next section,
we compare the network performances achieved by topology formation strategies including
the proposed algorithm.


\subsection{Performance Comparison}
In this section, we evaluate the performance of the proposed algorithm in terms of
network utility and computational complexity. The performance of the proposed algorithm is 
compared with the existing network formation
strategies shown below. 
\begin{enumerate}
\item \emph{Non-NC Centralized}: A centralized solution to find  an optimal 
  network topology based on exhaustive search. The network does not deploy network coding such that the solution should
  be found by explicitly considering the inter-link
dependency condition. This strategy can be formulated as 
\begin{align*}
&(a_i^*, a_{-i}^*) = \arg\max_{(a_i, a_{-i})} U(a_i, a_{-i}, \mathbf D)\\
&\text{subject to } \sum_{\forall v_j \in \mathcal H_i} E_{ij} (a_i) \le 1, \forall v_i \in \mathcal V(\mathcal G).
\end{align*}

\item \emph{NC Centralized}: A centralized solution to find  an optimal 
  network topology based on exhaustive search.  The network deploys  network coding such that the inter-link dependency
  condition cannot be considered. It can be formulated as
\begin{align*}
(a_i^*, a_{-i}^*) = \arg\max_{(a_i, a_{-i})} U(a_i, a_{-i}, \mathbf D).
\end{align*}

\item \emph{TCLE}~\cite{xu2016}: A distributed solution to topology formation based
  on a non-cooperative game. For fair comparison, we deploy network coding,  
  and no
  inter-link dependency condition is considered. 
In this strategy, a node chooses its transmission power by balancing the target network
connectivity redundancy $\epsilon$ against transmission energy dissipation. The number of
actions available for a node is denoted by $\eta$. 

\end{enumerate}

\begin{figure}[tb]
\centering
\includegraphics[width = 9cm]{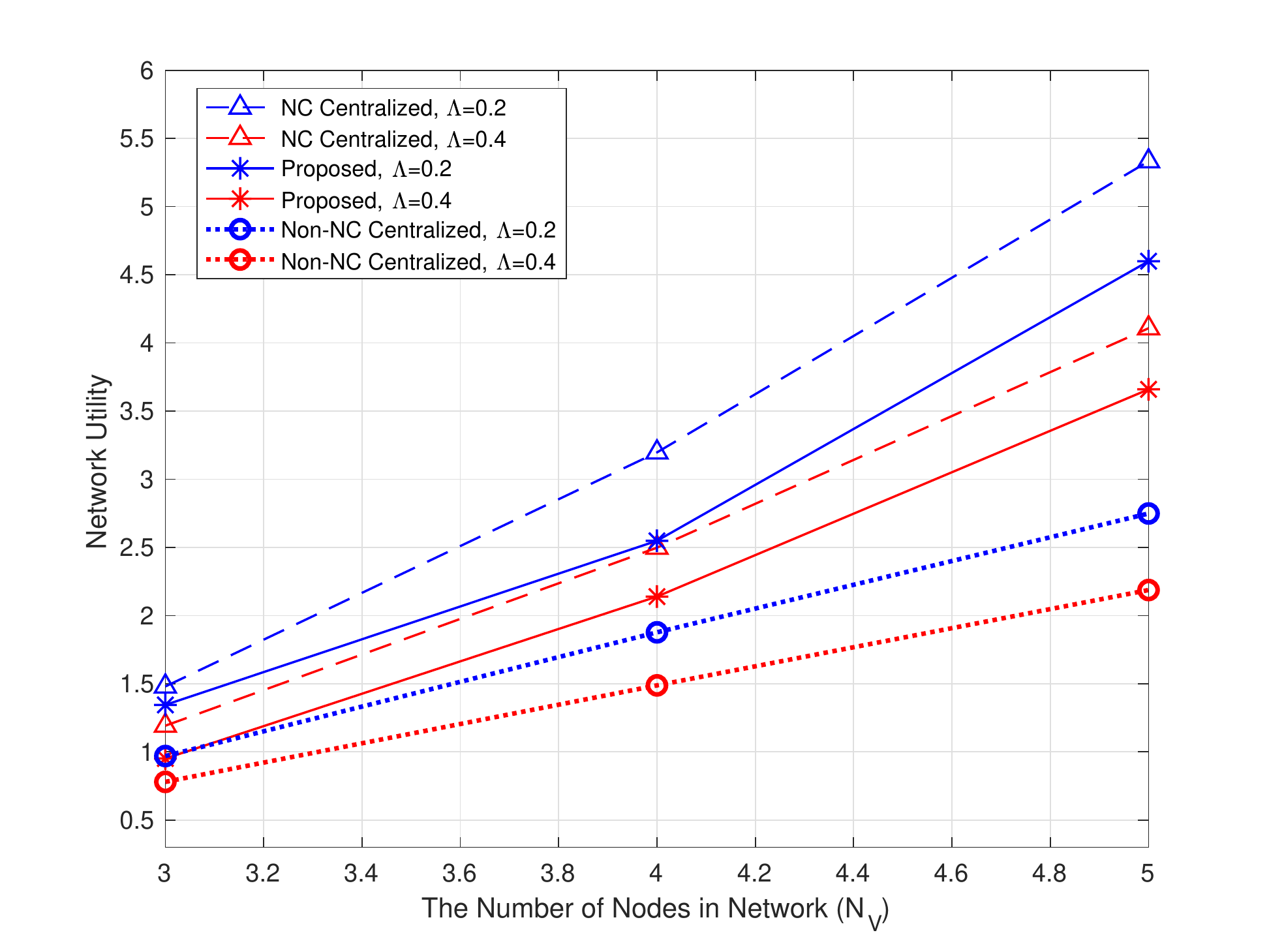}
\caption{Network utility for four strategies. 
}
\label{fig:comp_plot}
\end{figure}
Figure~\ref{fig:comp_plot} shows the network utility from each network formation
strategy, including the proposed algorithm. 
The NC-Centralized strategy can achieve the highest network utility
because 
the optimal network topology can be chosen from all possible topologies that can be
formed from other strategies. 
Hence, this can be considered as the upper bound of network coding based
strategies. 
The proposed strategy provides higher network utility than the Non-NC
Centralized strategy and TCLE. 
Unlike the proposed approach, which can consider more topologies to maximize rewards
by making 
multiple outgoing links, the Non-NC Centralized
strategy can find the optimal topology that is allowed only by 
the 
inter-link dependency condition. 
Among the considered strategies, TCLE provides the lowest network utility
because the focus of TCLE is not on construction of
successful connections between source nodes and 
destination nodes when it forms a network topology. Instead, it considers overall
network connectivity in terms of algebraic connectivity~\cite{gross2004}.

 \begin{figure}[tb]
\centering
\includegraphics[width = 9cm]{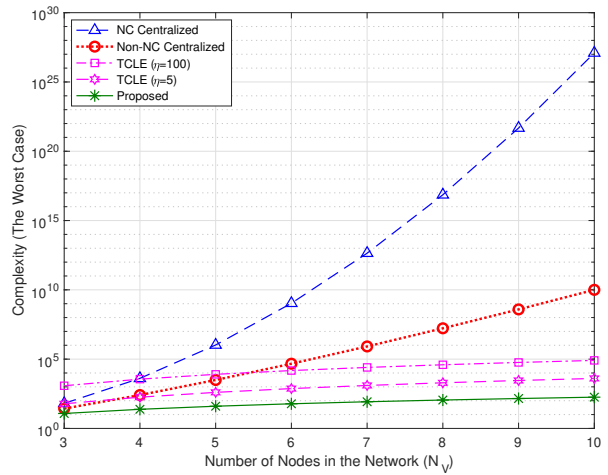}
\caption{Complexity required for the four strategies (worst case scenarios). 
}
\label{fig:complexity_comp}
\end{figure}

 \begin{figure}[tb]
\centering
\includegraphics[width = 9cm]{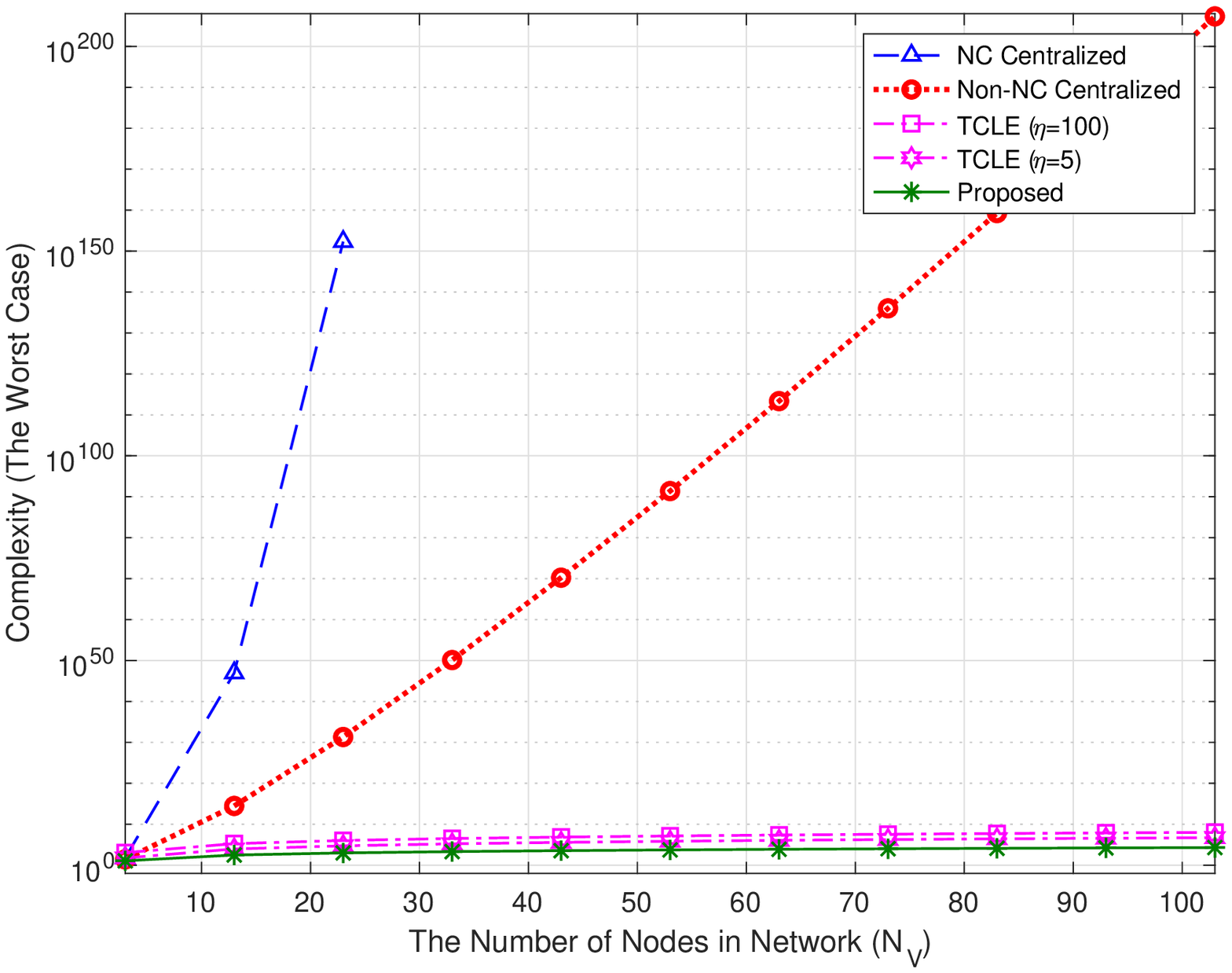}
\caption{Complexity required for the four strategies (worst case scenarios). 
}
\label{fig:complexity_comp_large}
\end{figure}

 \begin{figure}[tb]
\centering
\includegraphics[width = 9cm]{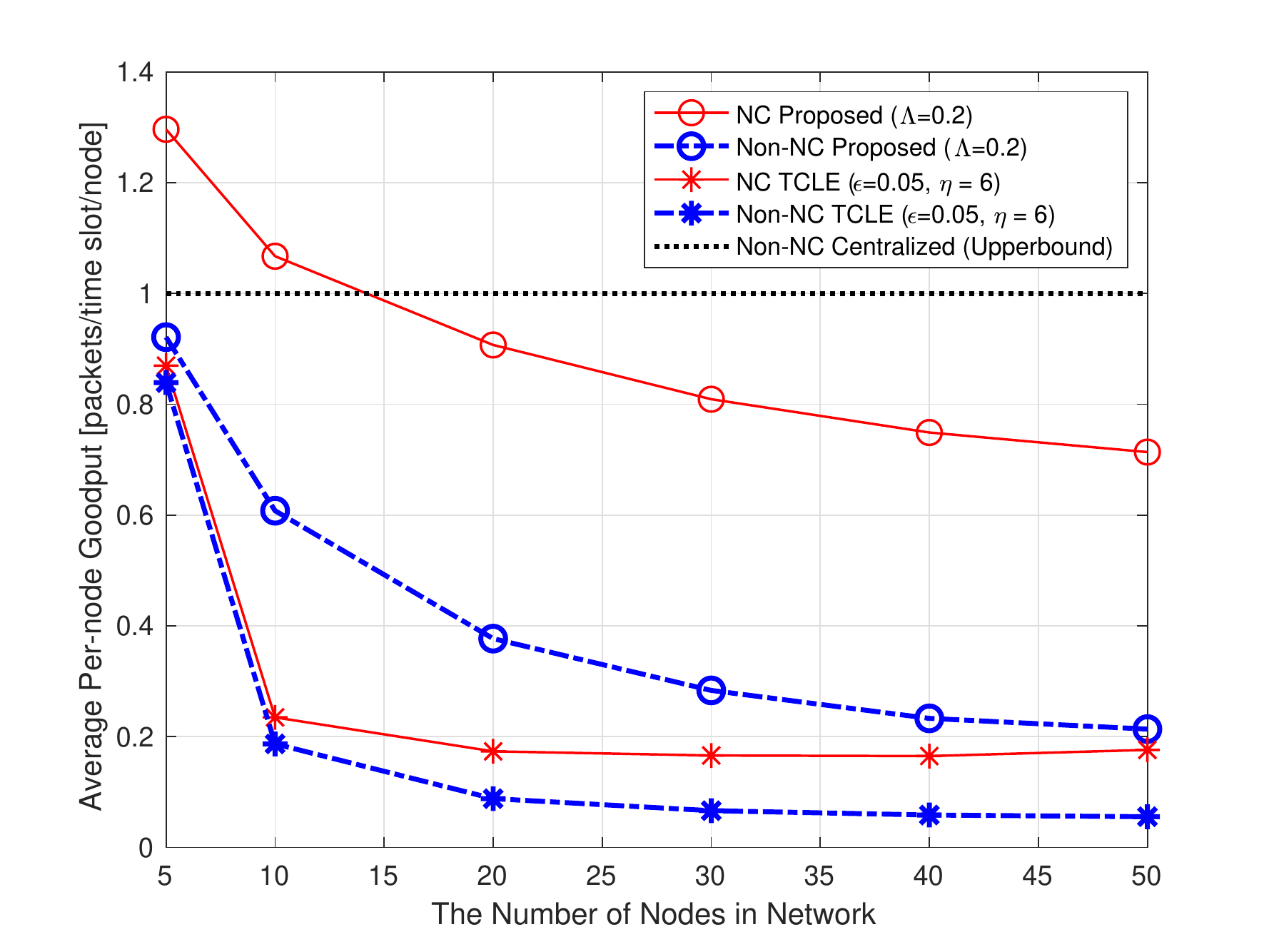}
\caption{Average per-node goodput for various network sizes and strategies
}
\label{fig:per-node goodput}
\end{figure}

\begin{table*}[tb]
\caption{ Theoretical complexity with network size $n$}
\centering
  \begin{tabular}{ |c||c | c | c| c| }
    \hline
Strategy &   Non-NC Centralized  & NC Centralized & TCLE \cite{xu2016} & Proposed \\ \hline \hline
  Complexity (Worst Case) & $\mathcal O(n^n)$ & $\mathcal O(2^{n^2})$ & $\mathcal O(\eta \cdot n^3)$ &$\mathcal O (n^2)$ \\
 \hline
  \end{tabular}
  \label{tab:complexity}
 \end{table*}

We next investigate the complexity required to deploy the network formation
strategies, which is  
measured by the size of the search space. The search space is determined by
the number of actions available for link formation. 
The size of the network considered in the complexity analysis is $n$. 
For the Non-NC Centralized strategy, each node can make at most one outgoing
link, so that a node has at most $n$ link formation choices. 
Given $n$ nodes in the network, the maximum number of choices is thus $n^n$, and the complexity is $\mathcal O(n^n)$. 
For the NC Centralized strategy, each node has a maximum $2^{n-1}$ choices because each node has a maximum of $n-1$ neighbor nodes to make a link, and each link
can be either active or inactive. 
Given $n$ nodes in the network, the maximum number of choices is thus $2^{n(n-1)}$ and the complexity becomes $\mathcal O(2^{n^2})$. 
The proposed algorithm has a maximum of $n \choose 2$ 
neighbor node pairs, and each pair has four actions. Because each node pair
chooses an action independently, the search space becomes ${n \choose 2} \times 4$, so that the complexity becomes $\mathcal O(n^2)$.
The complexity in TCLE  is $\mathcal O(\eta \cdot n^3)$, because there are $n$ nodes and each one has the worst case complexity $\mathcal O(\eta \cdot n^2)$. 
Therefore, the proposed strategy requires the lowest complexity to find the optimal
network topology, so that it can be deployed in practice into a large-scale network. 
The complexity required for the four network formation strategies is summarized in 
Table~\ref{tab:complexity} and presented in 
Figure~\ref{fig:complexity_comp} and Figure~\ref{fig:complexity_comp_large}.
In Figure~\ref{fig:complexity_comp_large},
it is clearly observed that the complexity for centralized strategies grows much 
faster than decentralized strategies as the network grows. Thus, it is not feasible
to use such high complexity centralized strategies in a realistic large network
setting. 



We further show the performance comparison in
terms of average per-node goodput (packets/time slot/node) for decentralized
solutions (i.e., proposed and TCLE) over increasing network sizes up to $50$ nodes. The
location of node is randomly generated with uniform distribution and this is the
average results from $4000$ independent simulations. The average per-node goodput is
defined as the average number of successfully delivered packets per time per node~\cite{nikolaidis2013building, brahma2012traffic, nawab2011tmac, pandya2008goodput} 
such that it is a fair measure to compare goodput of the network with
different sizes. The per-node goodput can be different upon utilization of network
coding. In Figure~\ref{fig:per-node goodput}, we show both cases for a given network - with and without network
coding.
Figure~\ref{fig:per-node goodput} shows that the proposed solution outperforms TCLE in all range of network
sizes and both NC and Non-NC cases. As the network size grows, it is observed that
the per-node goodput decreases across the strategies. This is because as the network
size increases, it requires a longer travel time to arrive at the destination,
leading to a lower per-node goodput. If each node can send maximum a packet per time
slot, which is the assumption of this simulation, the upper bound of per-node goodput
in an optimal network without network coding becomes 1 per-node goodput, i.e., all
nodes send 1 packet per time slot (fully utilizing the link capacity) and all
transmitted packets are successfully delivered to destination (no connection failure,
no packet loss).


\section{Conclusion}
\label{sec:conclusion}

In this paper, we propose an algorithm for a distributed topology formation in 
network coding
enabled large-scale sensor networks with multi-source multicast flows. 
The distributed topology formation problem is formulated as 
a network formation game in which the players (nodes in the network) decide
whether to make links with neighbor nodes by considering a reward for  
distance reduction and the cost required 
for link formation. 
We show that the network formation game can be decomposed into 
independent link formation games
by deploying network coding. Network topologies can thus be determined based on the
solution to individual link formation games played by only two nodes, leading to 
a distributed algorithm with low complexity. 
We also show that the proposed algorithm 
guarantees to find at least one network topology.  
The simulation results confirm that the proposed algorithm can achieve high network
utility with significantly low complexity and is scalable to any network
size. Therefore, it can be deployed 
in large-size networks. 

\appendix

\section{Network Coding Based Dissemination}
\label{sec:NC}

In the network operation in \eqref{eqn:y}, node $v_j$ combines its data $x_j$ and 
all the incoming data $X_{ij}, \forall v_i \in \tilde
{\mathcal  H}_j^{in}$ 
multiplied by local coding coefficients $c_{ij}, \forall v_i \in \{\tilde {\mathcal
H}_j^{in}, v_j\}$, expressed as 
\begin{align}
y_j&= \sum_{v_i \in \tilde {\mathcal  H}_j^{in}} \bigoplus \left( c_{ij} \otimes X_{ij} \right) \oplus c_{jj} \otimes x_j \notag\\
&= \sum_{v_i \in \tilde {\mathcal  H}_j^{in}} \bigoplus \left( c_{ij} \otimes y_{i} \right) \oplus c_{jj} \otimes x_j \label{eqn:encoding1}\\
&=\sum_{v_i \in \tilde {\mathcal  H}_j^{in}} \bigoplus \left( c_{ij} \otimes
\left(\sum_{k=1}^{N_V}\bigoplus  \left( C_{ki} \otimes x_k \right) \right) \right)
\oplus c_{jj} \otimes x_j. \label{eqn:encoding2}
\end{align}
In \eqref{eqn:encoding1}, $X_{ij}=y_i$ because $p_i= [C_{1i}, \ldots, C_{N_Vi}, y_i]$
is transmitted through $e_{ij}$, and
\eqref{eqn:encoding2} is induced from \eqref{eqn:y}. 

Since the global coding coefficient $C_{kj}$ is updated in every encoding process, 
\eqref{eqn:encoding2} can be expressed as 
\begin{align}
y_j&=\sum_{v_i \in \tilde {\mathcal  H}_j^{in}} \bigoplus   \left(\sum_{k=1}^{N_V} \bigoplus  c_{ij} \otimes  C_{ki}  \otimes x_k  \right) \oplus c_{jj} \otimes x_j \notag\\
&= \sum_{k=1}^{N_V} \bigoplus \left( \sum_{v_i \in \{ \tilde {\mathcal  H}_j^{in}, v_j \}} \bigoplus  c_{ij} \otimes  C_{ki}  
  \right) \otimes x_k \label{eqn:encoding4}\\
&=\sum_{k=1}^{N_V}\bigoplus  \left( C_{kj} \otimes x_k \right). \label{eqn:encoding5} 
\end{align}
In~\eqref{eqn:encoding5}, the global coding coefficient is updated based on $C_{ki} =\sum_{v_i \in \{ \tilde {\mathcal  H}_j^{in}, v_j \}} \bigoplus  c_{ij}
\otimes  C_{ki}  $. 
Thus, $p_j = [C_{1j}, \ldots, C_{N_Vj},
y_j]$ is constructed and forwarded to the nodes in $\tilde {\mathcal  H}_j^{out}$. 

In this paper, 
we assume that all data in a network are elements in a GF with size $2^M$, denoted as 
GF($2^M)$, i.e., $x_i, y_i \in$ GF($2^M)$,
as the network coding operations in \eqref{eqn:y} 
are performed in a GF. 
Moreover, we use random linear network coding (RLNC)~\cite{HO2006} 
so that the local coding
coefficient is randomly selected in GF($2^M)$.

Let 
$\mathcal S_i = \{ v_j| v_i \in \mathcal
D_j, \forall v_j \in \mathcal V(\mathcal G) \}$
be a set of source nodes 
whose destination set includes $v_i$ and 
$\mathbf S_i = \{ j | v_j \in \mathcal S_i \}$ be 
an index set of source nodes for $v_i$. 
Given the packets $\tilde p_1, \ldots, \tilde p_K$ that $v_i$ received, 
we can construct 
a vector of network coded data $\tilde {\mathbf y} =
[ \tilde y_1, \ldots, \tilde y_K  ]^T$ and 
the global coding coefficient
matrix $\tilde {\mathbf C}$, expressed as 
\begin{align}
\tilde {\mathbf C} 
&= 
\begin{bmatrix}
    \tilde {C}_{11} &  \cdots & \tilde {C}_{N_V 1} \\
    & \vdots & \\
    \tilde {C}_{1K} &  \cdots & \tilde {C}_{N_V K}
        \end{bmatrix}
        = [\tilde {\mathbf c}_1 \,\, \cdots \,\, \tilde{\mathbf c}_{N_V} ] 
\end{align}
where $\tilde{\mathbf c}_j = [\tilde {C}_{j1}, \ldots, \tilde {C}_{jK} ]^T$.

\section{Proof of Theorem~\ref{th:NE_proof}}
\label{sec:NE_proof}
Theorem~\ref{th:NE_proof} can be proved by showing that 
at least one pure strategy NE exists for the link formation game, 
i.e. $\mathbf A^* \ne \emptyset$.  

Nash's Existence Theorem~\cite{nash1950Eq} shows that every
finite game has a mixed strategy NE, where 
pure strategies are chosen stochastically
with certain probabilities. Since the link formation game
$\mathbf G_{\mathcal L}(d)$ defined in \eqref{eqn:lfg_def} 
includes a finite number of nodes
and a finite number of actions, it is a finite game. 
Therefore,  a
mixed strategy NE exists for this game. 

Suppose that $\alpha_i$ is a strategy of player $v_i$ with the probability
of taking action $a_i = 1$. The corresponding utility is given by
\begin{align*}
u_i(\alpha_i, \alpha_{j}, d) = \alpha_i \left(  f(\delta_{jd})
-f\left(\delta_{id}\right) -  \left( \frac{1}{1+E_{ji}(\alpha_{j})} \cdot \Lambda \right)
\right).
\end{align*}
Let $\alpha_i^*$ and $\alpha_j^*$ be mixed strategy NEs for $v_i$ and $v_j$ that satisfy 
$u_i (\alpha_i^*, \alpha_{j}^*,d) \ge u_i (\alpha_i, \alpha_{j}^*,d)$ 
and
$u_j (\alpha_i^*, \alpha_{j}^*,d) \ge u_j (\alpha_i^*, \alpha_{j},d)$
for all $v_i, v_j \in \mathcal V (\mathcal L)$.   

For the $\xi \in [-\alpha_i^*, 1-\alpha_i^*]$ perturbation of mixed strategy NE $\alpha_i^*$, the
resulting utility of $v_i$ can be expressed as
\begin{align*}
&u_i(\alpha_i^* + \xi, \alpha_{j}, d)\\
 &= (\alpha_i^* + \xi) \times \left(  f(\delta_{jd}) -f\left(\delta_{id}\right) -
 \left( \frac{1}{1+E_{ji}(\alpha_{j})} \cdot \Lambda \right) \right)\\
&=u_i(\alpha_i^* , \alpha_{j}, d) + \xi \times\left(  f(\delta_{jd})
-f\left(\delta_{id}\right) -  \left( \frac{1}{1+E_{ji}(\alpha_{j})} \cdot \Lambda \right) \right)\\
&=u_i(\alpha_i^* , \alpha_{j}, d) + \xi \times \beta.
\end{align*}
If $\beta > 0$,  then ${\partial u_i(\alpha_i^* + \xi, \alpha_{j}, d) \over \partial
\xi}
> 0$, and thus, $v_i$ can always decrease its utility by decreasing $\xi$. This 
means that the pure strategy $\alpha_{i}=0$ strictly dominates any strategies
$\alpha_{i}^*+\xi$.
For $\beta < 0$, on the other hand, ${\partial u_i(\alpha_i^* + \xi, \alpha_{j}, d) \over
\partial \xi} < 0 $, and thus, the pure strategy $\alpha_{i}=1$ strictly dominates  
any strategies $\alpha_{i}^*+\xi$.
If $\beta = 0$, then ${\partial u_i(\alpha_i^* + \xi, \alpha_{j}, d) \over \partial
\xi} = 0 $ and $\xi$ does not affect the cost such that both pure strategies
$\alpha_{i}=0$ and $\alpha_{i}=1$ have the same utility as the mixed strategy NE $\alpha^*_{i}$. 
Therefore, the pure strategy $\alpha_i=0$ or $\alpha_i=1$ can always weakly dominate mixed strategies $\alpha_i^*+\xi$. 

Similarly, the $\xi$ perturbation of mixed strategy NE $\alpha_j^*$ also concludes that the pure strategy $\alpha_j=0$ or $\alpha_j=1$ can always weakly dominate mixed strategies $\alpha_j^*+\xi$. 

In conclusion, 
a pure strategy NE can always weakly dominate mixed strategy NEs in the link
formation game
$\mathbf G_{\mathcal L}(d)$, implying that 
a pure strategy NE exists.  
Hence, 
Algorithm~\ref{alg:bigpicture} guarantees at least one topology. 

\section*{References}
\bibliography{mybibfile_v11}

\begin{thebibliography}{10}
\expandafter\ifx\csname url\endcsname\relax
  \def\url#1{\texttt{#1}}\fi
\expandafter\ifx\csname urlprefix\endcsname\relax\def\urlprefix{URL }\fi
\expandafter\ifx\csname href\endcsname\relax
  \def\href#1#2{#2} \def\path#1{#1}\fi

\bibitem{wifison}
\href{https://www.qualcomm.com/products/features/wi-fi-son}{Qualcomm {W}i-{F}i
  {SON} and distributed networking}.
\newline\urlprefix\url{https://www.qualcomm.com/products/features/wi-fi-son}

\bibitem{li2013blind}
Y.~Li, Z.~Zhang, C.~Wang, W.~Zhao, H.-H. Chen, Blind cooperative communications
  for multihop ad hoc wireless networks, IEEE Transactions on Vehicular
  Technology 62~(7) (2013) 3110--3122.

\bibitem{li2015energy}
Y.~Li, C.~Liao, Y.~Wang, C.~Wang, Energy-efficient optimal relay selection in
  cooperative cellular networks based on double auction, IEEE Transactions on
  Wireless Communications 14~(8) (2015) 4093--4104.

\bibitem{kar2008}
S.~Kar, J.~M. Moura, Sensor networks with random links: Topology design for
  distributed consensus, IEEE Transactions on Signal Processing 56~(7) (2008)
  3315--3326.

\bibitem{li2014energy}
Y.~Li, X.~Zhu, C.~Liao, C.~Wang, B.~Cao, Energy efficiency maximization by
  jointly optimizing the positions and serving range of relay stations in
  cellular networks, IEEE Transactions on Vehicular Technology 64~(6) (2014)
  2551--2560.

\bibitem{li2018joint}
Y.~Li, J.~Liu, B.~Cao, C.~Wang, Joint optimization of radio and virtual machine
  resources with uncertain user demands in mobile cloud computing, IEEE
  Transactions on Multimedia 20~(9) (2018) 2427--2438.

\bibitem{mhkwon_ICC17}
M.~Kwon, H.~Park, Network coding-based distributed network formation game for
  multi-source multicast networks, in: IEEE International Conference on
  Communications, 2017.

\bibitem{mhkwonTMC2019}
M.~{Kwon}, H.~{Park}, Network coding based evolutionary network formation for
  dynamic wireless networks, IEEE Transactions on Mobile Computing 18~(6)
  (2019) 1316--1329.
\newblock \href {http://dx.doi.org/10.1109/TMC.2018.2861001}
  {\path{doi:10.1109/TMC.2018.2861001}}.

\bibitem{kim2006}
D.~S. Kim, Y.~J. Chung, Self-organization routing protocol supporting mobile
  nodes for wireless sensor network, in: International Multi-Symposiums on
  Computer and Computational Sciences, Vol.~2, 2006, pp. 622--626.

\bibitem{komali2008}
R.~S. Komali, A.~B. MacKenzie, R.~P. Gilles, Effect of selfish node behavior on
  efficient topology design, IEEE Transactions on Mobile Computing 7~(9) (2008)
  1057--1070.

\bibitem{xu2016}
M.~Xu, Q.~Yang, K.~S. Kwak, Distributed topology control with lifetime
  extension based on non-cooperative game for wireless sensor networks, IEEE
  Sensors Journal 16~(9) (2016) 3332--3342.

\bibitem{chong2003}
C.-Y. Chong, S.~P. Kumar, Sensor networks: evolution, opportunities, and
  challenges, Proceedings of the IEEE 91~(8) (2003) 1247--1256.

\bibitem{kwonSPL2017}
M.~{Kwon}, H.~{Park}, Distributed network formation strategy for network coding
  based wireless networks, IEEE Signal Processing Letters 24~(4) (2017)
  432--436.

\bibitem{Aalamifar2018}
F.~Aalamifar, L.~Lampe, Cost-efficient qos-aware data acquisition point
  placement for advanced metering infrastructure, IEEE Transactions on
  Communications (2018) 1--1.

\bibitem{Johari2018}
P.~Johari, J.~M. Jornet, Nanoscale optical wireless channel model for
  intra-body communications: Geometrical, time, and frequency domain analyses,
  IEEE Transactions on Communications 66~(4) (2018) 1579--1593.

\bibitem{Deligiannis2017}
N.~Deligiannis, J.~F.~C. Mota, E.~Zimos, M.~R.~D. Rodrigues, Heterogeneous
  networked data recovery from compressive measurements using a copula prior,
  IEEE Transactions on Communications 65~(12) (2017) 5333--5347.

\bibitem{bellman1953bottleneck}
R.~Bellman, Bottleneck problems and dynamic programming, Proceedings of the
  National Academy of Sciences of the United States of America 39~(9) (1953)
  947.

\bibitem{berman1990constrained}
O.~Berman, D.~Einav, G.~Handler, The constrained bottleneck problem in
  networks, Operations Research 38~(1) (1990) 178--181.

\bibitem{punnen1996fast}
A.~P. Punnen, A fast algorithm for a class of bottleneck problems, Computing
  56~(4) (1996) 397--401.

\bibitem{Ahlswede2000}
R.~Ahlswede, N.~Cai, S.-Y.~R. Li, R.~W. Yeung, Network information flow, IEEE
  Transactions on Information Theory 46~(4) (2000) 1204--1216.

\bibitem{chi2008}
K.~{Chi}, X.~{Jiang}, S.~{Horiguchi}, M.~{Guo}, Topology design of
  network-coding-based multicast networks, IEEE Transactions on Parallel and
  Distributed Systems 19~(5) (2008) 627--640.
\newblock \href {http://dx.doi.org/10.1109/TPDS.2007.70743}
  {\path{doi:10.1109/TPDS.2007.70743}}.

\bibitem{Prior2014}
R.~Prior, D.~E. Lucani, Y.~Phulpin, M.~Nistor, J.~Barros, Network coding
  protocols for smart grid communications, IEEE Transactions on Smart Grid
  5~(3) (2014) 1523--1531.

\bibitem{Greco2015}
C.~Greco, M.~Kieffer, C.~Adjih, Necorpia: Network coding with random
  packet-index assignment for mobile crowdsensing, in: IEEE International
  Conference on Communications, 2015, pp. 6338--6344.

\bibitem{cao2014stackelberg}
B.~Cao, L.~Qiao, Y.~Li, Stackelberg game theoretic approach for probabilistic
  network coding in retransmission mechanism, in: 2014 IEEE 8th International
  Symposium on Embedded Multicore/Manycore SoCs, IEEE, 2014, pp. 15--20.

\bibitem{li2014retransmission}
Y.~Li, Y.~Wu, B.~Cao, L.~Qiao, W.~Zhang, W.~Tang, Retransmission mechanism with
  probabilistic network coding in wireless networks, in: 2014 IEEE Global
  Communications Conference, IEEE, 2014, pp. 4502--4507.

\bibitem{mhKWON2014WCNC}
M.~Kwon, H.~Park, P.~Frossard, Compressed network coding: Overcome
  all-or-nothing problem in finite fields, in: IEEE Wireless Communications and
  Networking Conference, 2014, pp. 2851--2856.

\bibitem{Yang2008}
M.~Yang, Y.~Yang, A linear inter-session network coding scheme for multicast,
  in: Proc. IEEE International Symposium on Network Computing and Applications,
  2008, pp. 177--184.

\bibitem{Chou2007}
P.~A. Chou, Y.~Wu, Network coding for the internet and wireless networks, IEEE
  Signal Processing Magazine 24~(5) (2007) 77--85.
\newblock \href {http://dx.doi.org/10.1109/MSP.2007.904818}
  {\path{doi:10.1109/MSP.2007.904818}}.

\bibitem{Lehman2004}
A.~R. Lehman, E.~Lehman, Complexity classification of network information flow
  problems, in: ACM-SIAM Symposium on Discrete Algorithms, 2004, pp. 142--150.

\bibitem{Kim2009jsac}
Y.~Kim, G.~D. Veciana, Is rate adaptation beneficial for inter-session network
  coding?, IEEE Journal on Selected Areas in Communications 27~(5) (2009)
  635--646.

\bibitem{Khreishah2010}
A.~Khreishah, C.~C. Wang, N.~B. Shroff, Rate control with pairwise intersession
  network coding, IEEE/ACM Transactions on Networking 18~(3) (2010) 816--829.

\bibitem{Bourtsoulatze2014TCOM}
E.~Bourtsoulatze, N.~Thomos, P.~Frossard, Decoding delay minimization in
  inter-session network coding, IEEE Transactions on Communications 62~(6)
  (2014) 1944--1957.

\bibitem{Hulya2009}
H.~Seferoglu, A.~Markopoulou, Distributed rate control for video streaming over
  wireless networks with intersession network coding, in: International Packet
  Video Workshop, 2009.

\bibitem{Bourtsoulatze2014TMM}
E.~Bourtsoulatze, N.~Thomos, P.~Frossard, Distributed rate allocation in
  inter-session network coding, IEEE Transactions on Multimedia 16~(6) (2014)
  1752--1765.

\bibitem{Douik2016}
A.~Douik, S.~Sorour, T.~Y. Al-Naffouri, M.~S. Alouini, Decoding delay
  controlled completion time reduction in instantly decodable network coding,
  IEEE Transactions on Vehicular Technology 66~(3) (2017) 2756 -- 2770.

\bibitem{PracticalNC03}
P.~A. Chou, Y.~Wu, K.~Jain, Practical network coding, in: Annual Allerton
  Conference on Communication, Control, and Computing, Monticell, IL, USA,
  2003.

\bibitem{jain2003}
K.~Jain, M.~Mahdian, M.~R. Salavatipour, Packing steiner trees, in: ACM-SIAM
  symposium on Discrete algorithms, 2003, pp. 266--274.

\bibitem{Li2003}
S.-Y.~R. Li, R.~W. Yeung, N.~Cai, Linear network coding, IEEE Transactions on
  Information Theory 49~(2) (2003) 371--381.

\bibitem{Cui2007}
Y.~Cui, Y.~Xue, K.~Nahrstedt, Optimal distributed multicast routing using
  network coding, in: IEEE International Conference on Communications, 2007,
  pp. 3610--3615.

\bibitem{Yan2006}
X.~Yan, J.~Yang, Z.~Zhang, An outer bound for multisource multisink network
  coding with minimum cost consideration, IEEE Transactions on Information
  Theory 52~(6) (2006) 2373--2385.

\bibitem{Traskov2006}
D.~Traskov, N.~Ratnakar, D.~S. Lun, R.~Koetter, M.~M\'{e}dard, Network coding
  for multiple unicasts: An approach based on linear optimization, in: IEEE
  International Symposium on Information Theory, 2006, pp. 1758--1762.

\bibitem{movassaghi2013}
S.~Movassaghi, M.~Shirvanimoghaddam, M.~Abolhasan, D.~Smith, An energy
  efficient network coding approach for wireless body area networks, in: IEEE
  Conference on Local Computer Networks, 2013, pp. 468--475.

\bibitem{sohrabi2000}
K.~Sohrabi, J.~Gao, V.~Ailawadhi, G.~J. Pottie, Protocols for self-organization
  of a wireless sensor network, IEEE Personal Communications 7~(5) (2000)
  16--27.

\bibitem{Heo2005}
N.~Heo, P.~K. Varshney, Energy-efficient deployment of intelligent mobile
  sensor networks, IEEE Transactions on Systems, Man, and Cybernetics - Part A:
  Systems and Humans 35~(1) (2005) 78--92.

\bibitem{song2015}
Y.~Song, M.~van~der Schaar, Dynamic network formation with incomplete
  information, Economic Theory 59~(2) (2015) 301--331.

\bibitem{eidenbenz2006}
S.~Eidenbenz, V.~Kumar, S.~Zust, Equilibria in topology control games for ad
  hoc networks, Mobile Networks and Applications 11~(2) (2006) 143--159.

\bibitem{katti2007joint}
S.~Katti, I.~Maric, A.~Goldsmith, D.~Katabi, M.~Medard, Joint relaying and
  network coding in wireless networks, in: IEEE International Symposium on
  Information Theory, 2007, pp. 1101--1105.

\bibitem{baik2008network}
I.-J. Baik, S.-Y. Chung, Network coding for two-way relay channels using
  lattices, in: IEEE International Conference on Communications, 2008, pp.
  3898--3902.

\bibitem{li2010relay}
Y.~Li, R.~H. Louie, B.~Vucetic, Relay selection with network coding in two-way
  relay channels, IEEE Transactions on Vehicular Technology 59~(9) (2010)
  4489--4499.

\bibitem{Xu2011node}
Y.~{Xu}, X.~{Chen}, Node localization in wireless sensor network using dynamic
  distance prediction algorithm, in: IEEE International Instrumentation and
  Measurement Technology Conference, 2011, pp. 1--4.

\bibitem{Sunyong2017}
S.~Kim, S.~Y. Park, D.~Kwon, J.~Ham, Y.~Ko, H.~Lim, Two-hop distance estimation
  in wireless sensor networks, International Journal of Distributed Sensor
  Networks 13~(2) (2017) 1550147716689683.

\bibitem{Katti2006}
S.~Katti, H.~Rahul, H.~Wenjun, D.~Katabi, M.~M{\'e}dard, J.~Crowcroft, Xors in
  the air: practical wireless network coding, IEEE/ACM Transactions on
  Networking 16~(3) (2008) 497--510.

\bibitem{nad2004}
T.~Nad, A.~Krishnamurthy, Problems with network coding in overlay networks,
  Techinical Report, Yale University.

\bibitem{topakkaya2011}
H.~Topakkaya, Network coding for wireless and wired networks: Design,
  performance and achievable rates, Ph.D. thesis, Iowa State University (2011).

\bibitem{Cloud2012}
J.~Cloud, L.~M. Zeger, M.~M\'{e}dard, Mac centered cooperation - synergistic
  design of network coding, multi-packet reception, and improved fairness to
  increase network throughput, IEEE Journal on Selected Areas in Communications
  30~(2) (2012) 341--349.

\bibitem{mirrezaei2014}
S.~M. Mirrezaei, M.~Dosaranian-Moghadam, M.~Yazdanpanahei, Effect of network
  coding and multi-packet reception on point-to-multi-point broadcast networks,
  Wireless Personal Communications 79~(3) (2014) 1859--1891.

\bibitem{schal1994quadratic}
M.~Sch{\"a}l, On quadratic cost criteria for option hedging, Mathematics of
  operations research 19~(1) (1994) 121--131.

\bibitem{huang1998nonparametric}
Y.~Huang, T.~A. Louis, Nonparametric estimation of the joint distribution of
  survival time and mark variables, Biometrika 85~(4) (1998) 785--798.

\bibitem{pele2010quadratic}
O.~Pele, M.~Werman, The quadratic-chi histogram distance family, in: European
  conference on computer vision, Springer, 2010, pp. 749--762.

\bibitem{filipovic2013density}
D.~Filipovi{\'c}, E.~Mayerhofer, P.~Schneider, Density approximations for
  multivariate affine jump-diffusion processes, Journal of Econometrics 176~(2)
  (2013) 93--111.

\bibitem{zhang2011game}
Y.~Zhang, M.~Guizani, Game theory for wireless communications and networking,
  CRC press, 2011.

\bibitem{grigoras2007cost}
D.~Grigoras, M.~Riordan, Cost-effective mobile ad hoc networks management,
  Future Generation Computer Systems 23~(8) (2007) 990--996.

\bibitem{anshelevich2008price}
E.~Anshelevich, A.~Dasgupta, J.~Kleinberg, E.~Tardos, T.~Wexler,
  T.~Roughgarden, The price of stability for network design with fair cost
  allocation, SIAM Journal on Computing 38~(4) (2008) 1602--1623.

\bibitem{herzog1997}
S.~Herzog, S.~Shenker, D.~Estrin, Sharing the cost of multicast trees: an
  axiomatic analysis, IEEE/ACM Transactions on Networking 5~(6) (1997)
  847--860.

\bibitem{arcaute2009}
E.~Arcaute, R.~Johari, S.~Mannor, Network formation: Bilateral contracting and
  myopic dynamics, IEEE Transactions on Automatic Control 54~(8) (2009)
  1765--1778.

\bibitem{feigenbaum2001}
J.~Feigenbaum, C.~H. Papadimitriou, S.~Shenker, Sharing the cost of multicast
  transmissions, in: Journal of Computer and System Sciences, Vol.~63, 2001,
  pp. 21--41.

\bibitem{Agarwal2009}
U.~Agarwal, U.~P. Singh, Graph theory, Laxmi Publications, 2009.

\bibitem{gupta2000}
P.~Gupta, P.~R. Kumar, The capacity of wireless networks, IEEE Transactions on
  Information Theory 46~(2) (2000) 388--404.

\bibitem{gross2004}
J.~L. Gross, J.~Yellen, Handbook of graph theory, CRC press, 2004.

\bibitem{nikolaidis2013building}
I.~Nikolaidis, K.~Iniewski, Building Sensor Networks: From Design to
  Applications, CRC Press, 2013.

\bibitem{brahma2012traffic}
S.~Brahma, M.~Chatterjee, K.~Kwiat, P.~K. Varshney, Traffic management in
  wireless sensor networks: Decoupling congestion control and fairness,
  Computer Communications 35~(6) (2012) 670--681.

\bibitem{nawab2011tmac}
F.~Nawab, K.~Jamshaid, B.~Shihada, P.-H. Ho, {TMAC}: Timestamp-ordered {MAC}
  for {CSMA/CA} wireless mesh networks, in: 2011 Proceedings of 20th
  International Conference on Computer Communications and Networks (ICCCN),
  IEEE, 2011, pp. 1--6.

\bibitem{pandya2008goodput}
A.~Pandya, A.~Kansal, G.~Pottie, Goodput and delay in networks with controlled
  mobility, in: 2008 IEEE Aerospace Conference, IEEE, 2008, pp. 1--8.

\bibitem{HO2006}
M.~M{\'e}dard, R.~Koetter, D.~Karger, M.~Effros, J.~Shi, B.~Leong, A random
  linear network coding approach to multicast, IEEE Transactions on Information
  Theory 52~(10) (2006) 4413--4430.

\bibitem{nash1950Eq}
J.~F. Nash, Equilibrium points in n-person games, Proceedings of the National
  Academy of Sciences of the United States of America 36~(1) (1950) 48--49.

\end{thebibliography}


\end{document}